\documentclass[a4paper,USenglish,cleveref]{lipics-v2019}

\hideLIPIcs

\usepackage{algorithm}
\usepackage[noend]{algpseudocode}
\usepackage{xspace}

\newtheorem{fact}[theorem]{Fact}
\crefname{enumi}{Property}{Properties}
\hypersetup{colorlinks,linkcolor=blue,urlcolor=blue,citecolor=blue}

\newcommand{\e}{\mathrm{e}}
\newcommand{\R}{\ensuremath{R}}
\newcommand{\LG}{\ensuremath{L}\xspace}
\newcommand{\M}{\ensuremath{M_\mathrm{S}}\xspace}
\newcommand{\A}{\ensuremath{A}\xspace}
\newcommand{\ST}{\ensuremath{S_*}\xspace}
\newcommand{\MT}{\ensuremath{M_*}\xspace}
\newcommand{\MTI}[1]{\ensuremath{M_{#1}\xspace}}
\newcommand{\LT}{\ensuremath{L_+}\xspace}
\newcommand{\emptyBucket}{\ensuremath{E}\xspace}
\newcommand{\D}{\ensuremath{D}\xspace}
\newcommand{\Dext}{\ensuremath{D_\mathrm{ext}}\xspace}
\newcommand{\smallBoundary}{\ensuremath{\phi}}
\newcommand{\gconst}{\ensuremath{\xi_\mathrm{c}}}
\newcommand{\f}{\textsf{f}}
\newcommand{\g}{\ensuremath{\xi}}
\newcommand{\fintegral}{\textsf{F}}
\newcommand{\water}{\textsf{Q}}
\newcommand{\cutintegral}{\textsf{P}}
\newcommand{\B}{\textsf{b}}
\newcommand{\gain}{\textsf{g}}
\newcommand{\load}{\textsf{load}}
\newcommand{\thresh}{\textsf{thr}(\M)}
\newcommand{\stack}{\textsf{pile}}
\newcommand{\level}{\textsf{level}_*}
\newcommand{\vecv}{\mathbf{v}}
\newcommand{\mn}{\textsf{min}}
\newcommand{\mt}{\textsf{mt}(\D)}
\newcommand{\T}{\textsf{T}}
\newcommand{\Tmax}{\T^*}
\newcommand{\dd}{\mathrm{d}}
\newcommand{\ALG}{\textsc{Rta}\xspace}
\newcommand{\OPT}{\textsc{Opt}\xspace}
\newcommand{\DET}{\textsc{Det}\xspace}

\nolinenumbers

\title{An Optimal Algorithm for Online Multiple Knapsack}
\titlerunning{An Optimal Algorithm for Online Multiple Knapsack}

\author{Marcin Bienkowski}{Institute of Computer Science, University of Wroc{\l}aw, Poland}{marcin.bienkowski@cs.uni.wroc.pl}{https://orcid.org/0000-0002-2453-7772}{}
\author{Maciej Pacut}{Faculty of Computer Science, University of Vienna, Austria}{maciej.pacut@univie.ac.at}{https://orcid.org/0000-0002-6379-1490}{}
\author{Krzysztof Piecuch}{Institute of Computer Science, University of Wroc{\l}aw, Poland}{kpiecuch@cs.uni.wroc.pl}{}{}

\authorrunning{M. Bienkowski, M. Pacut, and K. Piecuch}
\Copyright{Marcin Bienkowski, Maciej Pacut, and Krzysztof Piecuch}
\funding{Research supported by Polish National Science Centre grants 2016/22/E/ST6/00499 and 2016/23/N/ST6/03412}
\relatedversion{A version without appendix was published in the proceedings of ICALP 2020.}

\begin{document}

\maketitle

\begin{abstract}
In the online multiple knapsack problem, an algorithm faces a stream of items,
and each item has to be either rejected or stored irrevocably in one of $n$ bins
(knapsacks) of equal size. The gain of an~algorithm is equal to the sum of sizes
of accepted items and the goal is to maximize the total gain.

So far, for this natural problem, the best solution was the $0.5$-competitive
algorithm \textsc{FirstFit} (the result holds for any $n \geq 2$). We present
the first algorithm that beats this ratio, achieving the competitive ratio of
$1/(1+\ln(2))-O(1/n) \approx 0.5906 - O(1/n)$. Our algorithm is deterministic
and optimal up to lower-order terms, as the upper bound of $1/(1+\ln(2))$ for
randomized solutions was given previously by Cygan et al.~[TOCS 2016]. 
\end{abstract}

%%%%%%%%%%%%%%%%%%%%%%%%%%%%%%%%%%%%%%%%%%%%%%%%%%%%%%%%%%%%%%%%%%%%%%%%%%%%%%%%%%%%%%%%%%%%%%%%%
%%%%%%%%%%%%%%%%%%%%%%%%%%%%%%%%%%%%%%%%%%%%%%%%%%%%%%%%%%%%%%%%%%%%%%%%%%%%%%%%%%%%%%%%%%%%%%%%%

\section{Introduction}
\label{sec:introduction}

Knapsack problems have been studied in theoretical computer science for
decades~\cite{Keller16,KePfPi04}. In particular, in the
\emph{multiple knapsack} 
problem~\cite{ABFLNE02,BoFaLN01,CheKha05,CyJeSg16,IwaTak02,Keller99,Pising99},
items of given sizes and profits have to be stored in $n$ bins (knapsacks), each
of capacity $1$. The goal is to find a subset of all items that maximizes the
total profit and can be feasibly packed into bins without exceeding their
capacities. We consider an \emph{online scenario}, where an online algorithm is
given a sequence of items of unknown length. When an~item is presented to an
algorithm, it has to either irrevocably reject the item or accept it to a~chosen
bin (which cannot be changed in the future). The actions of an online algorithm
have to be made without the knowledge of future items.

\subparagraph{Proportional case.}

In this paper, we focus on the most natural, \emph{proportional variant}
(sometimes called \emph{uniform}), where item profits are equal to item sizes
and the goal is to maximize the sum of profits of all accepted items.

The single-bin case ($n = 1$) has been fully resolved: no deterministic online
algorithm can be competitive~\cite{MarVer95}, and the best randomized algorithm
\textsc{ROne} by Böckenhauer et al.~\cite{BKoKR14} achieves the optimal
competitive ratio of $0.5$.\footnote{An online algorithm is called
$\alpha$-competitive if, for any input instance, its total profit is at least
fraction~$\alpha$ of the optimal (offline) solution. While many papers use the
reciprocal of $\alpha$ as the competitive ratio, the current definition is more
suited for accounting arguments in our proofs.} 

Less is known for multiple-bin case ($n \geq 2$). Cygan et al.~\cite{CyJeSg16}
showed that the \textsc{FirstFit} algorithm is $0.5$-competitive and proved that
no algorithm (even a~randomized one) can achieve a competitive ratio greater than
$\R$, where 
\[
	\R = 1/(1+\ln 2) \approx 0.5906.
\]

\subparagraph{Other variants.}

Some authors focused on the variant, where the goal is to maximize the
\emph{maximum} profit over all bins, instead of the sum of the profits. For this
objective, optimal competitive ratios are already known: $0.5$-competitive
deterministic algorithm was given by Böckenhauer~et~al.~\cite{BKoKR14}, and the
upper bound of $0.5$ holding even for randomized solutions were presented by
Cygan~et~al.~\cite{CyJeSg16}.

The multiple knapsack problem can be generalized in another direction: profits
and sizes may be unrelated. However, already \emph{the unit variant}, where the
profit of each item is equal to $1$, does not admit any competitive solutions
(even randomized ones)~\cite{BoFaLN01}. 

These results together mean that the proportional case studied in this paper is
the only variant, whose online complexity has not been fully resolved yet.

%%%%%%%%%%%%%%%%%%%%%%%%%%%%%%%%%%%%%%%%%%%%%%%%%%%%%%%%%%%%%%%%%%%%%%%%%%%%%%%

\subsection{Our results}

The main result of this paper is an $(\R-O(1/n))$-competitive deterministic
online algorithm for the proportional variant of the multiple knapsack problem.
We give insights for our construction in \cref{sec:idea} below and the
definition of our algorithm later in \cref{sec:algorithm}. Given the upper bound
of $\R$ for randomized solutions~\cite{CyJeSg16}, our result is optimal up to
lower-order terms also for the class of randomized solutions.

It is possible to show that for deterministic algorithms, 
the term $O(1/n)$ in the competitive 
ratio is inevitable: we show how the upper bound
construction given in~\cite{CyJeSg16} can be tweaked and extended to show that the
competitive ratio of any deterministic algorithm is at most $\R-O(1/n)$.

%%%%%%%%%%%%%%%%%%%%%%%%%%%%%%%%%%%%%%%%%%%%%%%%%%%%%%%%%%%%%%%%%%%%%%%%%%%%%%%

\subsection{Related work}

Some previous papers focused on a \emph{removable} scenario, where an accepted item can
be removed afterwards from its
bin~\cite{ABFLNE02,CyJeSg16,HaKaMa15,IwaTak02,IwaZha10}. Achievable competitive
ratios are better than their non-removable counterparts; in particular, the
proportional variant admits constant-competitive deterministic algorithms even
for a single bin~\cite{IwaTak02}.

The online knapsack problem has been also considered in relaxed variants: with
resource augmentation, where the bin capacities of an online algorithm are
larger than those of the optimal offline one~\cite{IwaZha10,NogSar05}, with a
resource buffer~\cite{HaKaMY19}, or in the variant where an algorithm may accept
fractions of items~\cite{NogSar05}.

The hardness of the variants with arbitrary profits and sizes as well as
applications to online auctions motivated another strand of research focused on
the so-called random-order model~\cite{AlKhLa19,BaImKK07,KeRaTV18,Vaze17}.
There, the set of items is chosen adversarially, but the items are presented to an
online algorithm in random order.

%%%%%%%%%%%%%%%%%%%%%%%%%%%%%%%%%%%%%%%%%%%%%%%%%%%%%%%%%%%%%%%%%%%%%%%%%%%%%%%

\subsection{Algorithmic challenges and ideas}
\label{sec:idea}

Our algorithm splits items into three categories: large (of size greater than
$1/2$), medium (of size from the interval $[\smallBoundary, 1/2]$) and small (of size
smaller than $\smallBoundary$). We defer the actual definition of $\smallBoundary$. 

First, we explain what an online algorithm should do when it faces a stream of
large items. Note that no two large items can fit together in a single bin. If
an algorithm greedily collects all large items, then the adversary may give $n$
items of size $1/2+\epsilon$ (accepted by an~online algorithm) followed by $n$
items of size $1$ (accepted by an optimal offline algorithm \OPT), and the 
resulting competitive ratio is then $0.5$. On the other
hand, if an algorithm stops after accepting some number of large items, 
\OPT may collect all of them. 

Our \textsc{Rising Threshold Algorithm} (\ALG) balances these two strategies. It
chooses a~non-decreasing threshold function $\f: [0,1] \to [1/2,1]$ and ensures
that the size of the $i$-th accepted large item is at least $\f(i/n)$. While an
actual definition of $\f$ is given later, to grasp a~general idea, it is worth
looking at its plot in \cref{fig:def-fpq} (left). A natural adversarial strategy
is to give large items meeting these thresholds, and once \ALG fills $k$ bins,
present $n$ items of sizes slightly smaller than the next threshold
$\f((k+1)/n)$. These items will be rejected by \ALG but can be accepted by \OPT.
Analyzing this strategy and ensuring that the ratio is at least~$\R$ for any
choice of $k$ yields boundary conditions. Analyzing these conditions for 
$n$~tending to infinity, we obtain a differential equation, whose solution is the
function $\f$ used in our algorithm.

The actual difficulty, however, is posed by medium items. \ALG never proactively
rejects them and it keeps a subset of \emph{marked} medium items in their own
bins (one item per one bin), while it stacks the remaining, non-marked ones
(places them together in the same bin, possibly combining items of similar
sizes). This strategy allows \ALG to combine a large item with marked medium
items later. However, the amount of marked items has to be carefully managed as
they do not contribute large gain alone. A typical approach would be to
partition medium items into discrete sub-classes, control the number of items in
each class, and analyze the gain on the basis of the \emph{minimal} size item in
a particular subclass. To achieve optimal competitive ratio, we however need a
more fine-grained approach: we use a~carefully crafted continuous function $\g$
to control the number of marked items larger than a~given value. Analyzing all
possible adversarial strategies gives boundary conditions for $\g$. In
particular, the value $\smallBoundary$ that separates medium items from small
ones was chosen as the minimum value that ensures the existence of function $\g$
satisfying all boundary conditions.

Finally, we note that simply stacking small items in their own bins would not
lead to the desired competitive ratio. Instead, \ALG tries to stack them in a
single bin, but whenever its load exceeds $\smallBoundary$, \ALG
tries to merge them into a single medium item and verify whether such an item
could be marked. This allows for combining them in critical cases with large
items.

%%%%%%%%%%%%%%%%%%%%%%%%%%%%%%%%%%%%%%%%%%%%%%%%%%%%%%%%%%%%%%%%%%%%%%%%%%%%%%%

\subsection{Preliminaries}

We have $n$ bins of capacity $1$, numbered from $1$ to $n$. An input is a stream
of items from $(0, 1]$, defined by their sizes. Upon seeing an item, an online
algorithm has to either reject it or place it in an arbitrary bin without
violating the bin's capacity. The \emph{load} of a~bin $b$, denoted $\load(b)$,
is the sum of item sizes stored in bin $b$. We define the load of a set of items
as the sum of their sizes and the total load as the load of all items collected
by an algorithm. Additionally, for any $x \leq 1/2$, we define 
$\stack(x) = \max\{2/3, 2 x\}$. Note that if we put medium items of sizes 
at least $x$ (till it is possible) into a bin $b$, then $\load(b) \geq \stack(x)$.

To simplify calculations, for any set $Z$ of items, we define the \emph{gain} of
$Z$, denoted $\gain(Z)$, as their load divided by~$n$; similarly, the total gain
is the total load divided by~$n$. Furthermore, we use $\mn(Z)$ to denote the
minimum size of an item in set $Z$. If $Z$ is accepted by our online algorithm,
$\B(Z)$ denotes the number of bins our algorithm uses to accommodate these items,
divided by $n$. For any value $x \in [0,1]$, $Z^{\geq x}$ is the set of
all items from $Z$ of size greater or equal $x$. Whenever we use terms 
$\gain(Z)$, $\mn(Z)$ or $\B(Z)$ for a set $Z$ that varies during runtime, we mean these values
for the set~$Z$ after an online algorithm terminates its execution.

For any input sequence $\sigma$ and an algorithm $A$, we use $A(\sigma)$ to denote
the total gain of $A$ on sequence $\sigma$. We denote the optimal offline
algorithm by $\OPT$.

%%%%%%%%%%%%%%%%%%%%%%%%%%%%%%%%%%%%%%%%%%%%%%%%%%%%%%%%%%%%%%%%%%%%%%%%%%%%%%%

\subsection{Neglecting lower-order terms}

As our goal is to show the competitive ratio $\R - O(1/n)$, we introduce a
notation that allows to neglect terms of order~$1/n$. We say that $x$ is
\emph{approximately equal} to $y$ (we write $x \eqsim y$) if $|x - y| = O(1/n)$.
Furthermore, we say that $x$ is \emph{approximately greater} than $y$ (we write
$x \gtrsim y$) if $x \geq y$ or $x \eqsim y$; we define relation $\lesssim$
analogously. Each of these relations is transitive when composed a
constant number of times.

In our analysis, we are dealing with Lipschitz continuous functions 
(their derivative is bounded by a universal constant). For such function $h$, 
(i) the relation $\eqsim$ is preserved after application of $h$, and (ii) 
an integral of $h$ can be approximated by a~sum, as stated in the following facts,
used extensively in the paper.

\begin{fact}
\label{lem:Lipschitz}
Fix any Lipschitz continuous function $h$ and values $x \eqsim y$ from its domain. 
Then, $h(x) \eqsim h(y)$. Furthermore, if $h$ is non-decreasing, then
$x \lesssim y$ implies $h(x) \lesssim h(y)$
and $x \gtrsim y$ implies $h(x) \gtrsim h(y)$.
\end{fact}

\begin{fact}
\label{lem:sum_int}
For any Lipschitz continuous function $h$ and integers $a$, $b$ satisfying $1
\leq a \leq b \leq n$, it holds that $(1/n) \cdot \sum_{i=a+1}^b h(i/n) \eqsim
\int_{a/n}^{b/n} h(x)\, \dd x$.
\end{fact}

%%%%%%%%%%%%%%%%%%%%%%%%%%%%%%%%%%%%%%%%%%%%%%%%%%%%%%%%%%%%%%%%%%%%%%%%%%%%%%%

\subsection{Roadmap of the proof}

We present our algorithm in \cref{sec:algorithm}. 
Its analysis consists of three main parts. 
\begin{itemize}
\item In \cref{sec:large1}, we investigate the gain of \ALG on large items and explain
the choice of the threshold function $\f$. 
\item In \cref{sec:medium}, we study properties of medium items, marking 
routine, function $\g$ and show how the marked items influence the gain 
on other non-large items. 
\item In \cref{sec:large2}, we study the impact of marked items on 
bins containing large items. 
\end{itemize}
Each of these parts is concluded with a statement that, under certain
conditions, \ALG is $(R-O(1/n))$-competitive (cf.~\cref{lem:empty},
\cref{lem:no-m-bin}, \cref{lem:small-ms} and \cref{lem:large-ms}). In
\cref{sec:final}, we argue that these lemmas cover all possible outcomes. For
succinctness, some technical claims have been moved to \cref{sec:proofs}.

%%%%%%%%%%%%%%%%%%%%%%%%%%%%%%%%%%%%%%%%%%%%%%%%%%%%%%%%%%%%%%%%%%%%%%%%%%%%%%%%%%%%%%%%%%%%%%%%%
%%%%%%%%%%%%%%%%%%%%%%%%%%%%%%%%%%%%%%%%%%%%%%%%%%%%%%%%%%%%%%%%%%%%%%%%%%%%%%%%%%%%%%%%%%%%%%%%%

\section{Rising Threshold Algorithm}
\label{sec:algorithm}

We arrange items into three categories: small, medium and large. We say that an
item is \emph{large} if its size is in the range $(1/2,1 ]$, \emph{medium} if it
is in the range $[\smallBoundary, 1/2]$, and otherwise it is \emph{small},
where we define
\begin{align}
  \label{eq:gconst_def}
  \gconst & = (1 + (2/3) \cdot \ln(4/3)) \cdot \R - 2/3 \approx 0.0372 \quad \text{and}  \\
  \label{eq:small_boundary}
  \smallBoundary & = (2/3) \cdot \gconst \,/\, (2/3 - \R + \gconst) \approx 0.2191.
\end{align}
We further arrange medium items into subcategories $\MTI{2}$, $\MTI{3}$ and
$\MTI{4}$: a medium item belongs to $\MTI{i}$ if its size is from range
$(1/(i+1), 1/i]$. As we partition only medium items this way, $\MTI{4}$~contains
items of sizes from $[\smallBoundary, 1/4]$. Note that at most $i$ items of
category~$\MTI{i}$ fit in a single bin. 

At some times (defined precisely later) a group of small items of a total load
from $[\smallBoundary, 2 \smallBoundary)$ stored in a single bin may
become \emph{merged}, and from that point is treated as a~single medium item. We
ensure that such merging action does not violate invariants of our algorithm.

Our algorithm \ALG applies labels to bins; the possible labels are \emptyBucket,
\A, \ST, \M, $\MTI{2}$, $\MTI{3}$, $\MTI{4}$ and~\LT. Each bin starts as an
\emptyBucket-bin, and \ALG can relabel it later. The label determines
the content of a given bin:
\begin{itemize}
  \item an $\emptyBucket$-bin is empty,
  \item an \A-bin (an \emph{auxiliary bin}) contains small items of a total load 
      smaller than $\smallBoundary$ and at most one \A-bin exists at any time,
  \item an \ST-bin contains one or multiple small items,
  \item an \M-bin contains a single marked medium item,
  \item an $\MTI{i}$-bin contains one or more medium items of
  category $\MTI{i}$,
  \item an \LT-bin contains a single large item and possibly some other non-large ones.
\end{itemize}
For any label $C$, we define a corresponding set, also denoted $C$, containing all items
stored in bins of label $C$. For instance, \LT is a set containing 
all items stored in \LT-bins.
Furthermore, we define \LG as the set of all large items (clearly $\LG \subseteq \LT$
and $\B(\LG) = \B(\LT)$)
and the set $\MT = \MTI{2} \uplus \MTI{3} \uplus \MTI{4}$.

%%%%%%%%%%%%%%%%%%%%%%%%%%%%%%%%%%%%%%%%%%%%%%%%%%%%%%%%%%%%%

\ALG processes a stream of items, and it operates until the stream ends or there
are no more empty bins (even if an incoming item could fit in some partially
filled bin). Upon the arrival of an item, \ALG classifies it by its size and
proceeds as described below. 

\subparagraph{Large items.}

Whenever a large item arrives, \ALG compares its size with the threshold
$\f(\B(\LG) + 1/n)$, and if the item is smaller, \ALG rejects it.
The function $\f: [0,1] \to [1/2,1]$ is defined as 
\begin{equation}
\label{eq:f-definition}
   \f(x) =
\begin{cases}
  1/2 &\mbox{if } x \leq \R, \\
  (2\e)^{x-1} & \textnormal{otherwise,}
\end{cases}
\end{equation}
and depicted in \cref{fig:def-fpq} (left).
If the item meets the threshold, \ALG attempts to put it in an \M-bin
with sufficient space left (relabeling it to \LT), and if no
such bin exists, \ALG puts the item in any empty bin.

\begin{figure}[t]
   \centering
   \includegraphics[width=0.75\textwidth]{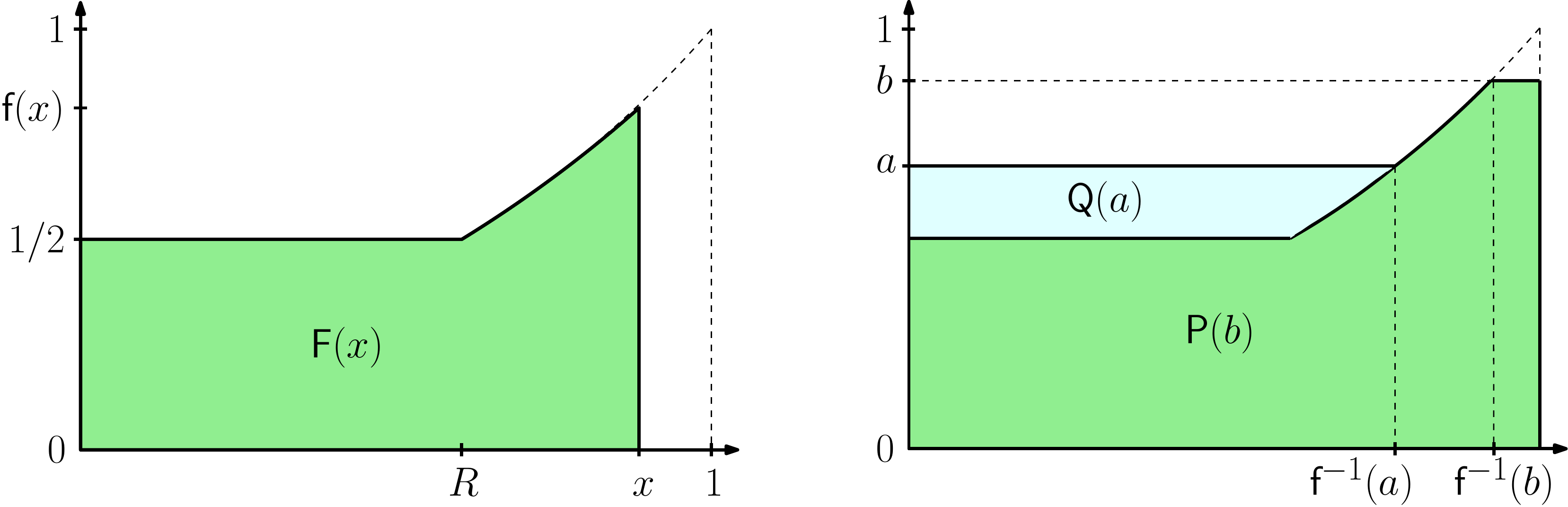}
   \caption{Left: function $\f$ and its integral $\fintegral$. 
   The value of $\fintegral(x)$ roughly corresponds to our lower bound on the 
   gain of \ALG when it collects $n \cdot x$ large items.
   Right: functions $\cutintegral$ and $\water$ used in estimating the gain 
   in \cref{sec:medium} and \cref{sec:large2}; note that their arguments are marked at Y axis.}
   \label{fig:def-fpq} 
\end{figure}

\subparagraph{Medium items.}

We fix a continuous and decreasing function $\g$ that maps medium item sizes 
to $[0,\gconst/\smallBoundary]$:
\begin{equation}
\label{eq:g-definition}
\g(x) =
\begin{cases}
  \gconst / x & \text{if $x \in [\smallBoundary,1/3]$}, \\
  9 \gconst \cdot (1-2 x) & \text{if $x \in (1/3,1/2]$}. 
\end{cases} 
\end{equation}
We say that the subset $Z$ of medium items is
$\g$-dominated if $|Z^{\geq x}| / n \leq \g(x)$ for any $x \in Z$.
Intuitively, it means that if we sort items of $Z$ from largest to
smallest, then all points~$(i/n, x_i)$ are under or at the plot of $\g^{-1}$,
see \cref{fig:tightness} (left). 

\ALG never proactively rejects medium items, i.e., it always accepts them if it 
has an~empty bin. Some medium items become \emph{marked} upon
arrival; we denote the set of medium marked items by $D$. Large or small items
are never marked. Marked medium items are never combined in a single bin with
other marked medium items. At all times, \ALG ensures that the set $D$ is
$\g$-dominated. As no two items from $D$ are stored in a single bin, this
corresponds to the condition $\B(D^{\geq x}) \leq \g(x)$ for any $x \in D$.
Each marked item is stored either in an $\M$-bin (alone) or in an $\LT$-bin (together 
with a large item and possibly some other non-marked items).
That is, $\M \subseteq \D \subseteq \M \uplus \LT$.

Whenever a medium item arrives, \ALG attempts to put it in an \LT-bin. If it
does not fit there, \ALG verifies whether marking it (including it in the set
$D$) preserves $\g$-domination of~$D$. If so, \ALG marks it and stores it in a
separate \M-bin. Otherwise, \ALG \emph{fails to mark} the item and the item is
stored in an $\MTI{i}$-bin (where $i$~depends on the item size): it is added to
an existing bin whenever possible and a new $\MTI{i}$-bin is opened only when
necessary.

We emphasize that if \ALG puts a large item in an \M-bin later (and relabel it
to \LT), the sole medium item from this bin remains marked (i.e., in the set $D$).
However, if a medium item fits in an \LT-bin at the time of its arrival, it
avoids being marked, even though its inclusion might not violate $\g$-dominance
of the set $D$. Note also that $\MT$ contains medium items \ALG failed to mark.

\subparagraph{Small items.}

\ALG never proactively rejects any small item. Whenever a small item arrives,
\ALG attempts to put this item in an \LT-bin, in an \ST-bin, and in the \A-bin,
in this exact order. If the item does not fit in any of them (this is possible
only if the \A-bin does not exist), \ALG places it in an empty bin and relabels
this bin to \A.

If \ALG places the small item in an already existing \A-bin and in effect its load
reaches or exceeds~$\smallBoundary$, \ALG attempts to \emph{merge} all its items into a single
medium marked item. If the resulting medium item can be marked
and included in~$D$ without violating its $\g$-dominance, 
\ALG relabels the \A-bin to \M and treats its contents as a single marked medium
item from now on. Otherwise, it simply changes the label of the \A-bin to \ST.

%%%%%%%%%%%%%%%%%%%%%%%%%%%%%%%%%%%%%%%%%%%%%%%%%%%%%%%%%%%%%%%%%%%%%%%%%%%%%%%%%%%%%%%%%%%%%%%%%
%%%%%%%%%%%%%%%%%%%%%%%%%%%%%%%%%%%%%%%%%%%%%%%%%%%%%%%%%%%%%%%%%%%%%%%%%%%%%%%%%%%%%%%%%%%%%%%%%

\section{Gain on large items}
\label{sec:large1}

In this section, we analyze the gain of \ALG on large items. To this end, we
first calculate the integral of function $\f$, denoted $\fintegral$
(see~\cref{fig:def-fpq}, left) and list its properties that can be verified by
routine calculations. 
\begin{equation}
\fintegral(x) = \int_0^x \f(y)\, dy =
  \begin{cases}
    x/2 & \mbox{if } x \leq \R,\\
    \R\cdot(2e)^{x-1} & \mbox{otherwise.}
  \end{cases}
  \label{eq:Ff}
\end{equation}

\begin{lemma}
\label{lem:F-properties}
The following properties hold for function $\fintegral$. 
\begin{enumerate}
  \item $\f^{-1}(c) = 1 + \R \cdot \ln c$ and 
  $\fintegral(\f^{-1}(c)) =  \R \cdot c$ 
  for any $c \in (1/2,1]$.
\item $\fintegral(x) / \f(x) = \min \{x,\R\}$ for any $x \in [0,1]$.
\item $(1/n) \cdot \sum_{i=1}^\ell \f(i/n) \eqsim \int_0^{\ell/n} \f(x)\, dx = \fintegral(\ell/n)$
for any $\ell \in \{0, \ldots, n\}$.
\end{enumerate}
\end{lemma}

Using \cref{lem:F-properties}, we may bound on the gain of \ALG on large items $L$
and use this bound to estimate its competitive ratio when it terminates with empty bins.

\begin{lemma}
\label{lem:large-gain}
It holds that $\gain(\LG) \gtrsim \fintegral(\B(\LG)) = \fintegral(\B(\LT))$. Moreover, $\gain(\LG)
\gtrsim \fintegral(\B(\LG)) + \gain(\LG^{\geq x}) - x \cdot \B(\LG^{\geq x})$
for any $x \geq \f(\B(\LG))$.
\end{lemma}

\begin{proof}
For the first part of the lemma, we sort large
items from $L$ in the order they were accepted by \ALG. The size of the $i$-th
large item is at least the threshold $f(i/n)$. Hence, by
\cref{lem:F-properties}, $\gain(L) \geq (1/n) \cdot \sum_{i=1}^{|L|} f(i/n)
\eqsim \fintegral(|L|/n) = \fintegral(\B(\LG))$.

To show the second part, we fix any $x \geq \f(\B(\LG))$ and for each large item
of size greater than~$x$ we reduce its size to~$x$. The total gain of the
removed parts is exactly $\gain(\LG^{\geq x}) - x \cdot \B(\LG^{\geq x})$. The
resulting large item sizes still satisfy acceptance thresholds, and thus the
gain on the remaining part of~$\LG$ is approximately greater than
$\fintegral(\B(\LG))$. Summing up yields $\gain(\LG) \gtrsim
\fintegral(\B(\LG)) + \gain(\LG^{\geq x}) - x \cdot \B(\LG^{\geq x})$.
\end{proof}

%%%%%%%%%%%%%%%%%%%%%%%%%%%%%%%%%%%%%%%%%%%%%%%%%%%%%%%%%%%%%%%%%%%%%%%%%%%%

\subsection{When RTA terminates with some empty bins}

\begin{lemma}
\label{lem:empty}
If \ALG terminates with some empty bins, then it is $(R - O(1/n))$-competitive. 
\end{lemma}

\begin{proof}
Fix an input sequence $\sigma$. As \ALG terminates with empty bins, it manages to 
accept all medium and small items from $\sigma$. Furthermore, it accepts large 
items from $\sigma$ according to the thresholds given by function $\f$.
Recall that $\f$ is non-decreasing: at the beginning it is equal to
$1/2$ (\ALG accepts any large item) and the acceptance threshold 
grows as \ALG accepts more large items. 
Let $x = \f(\B(\LG)+1/n)$ be the value of the acceptance
threshold for large items when $\ALG$ terminates.
We consider two cases. 

\begin{itemize}
\item $\B(\LG) \leq \R - 1/n$.
The threshold used for each large item is at most $x \leq \f(\R) = 1/2$, i.e.,
\ALG accepts all large items. Then, \ALG accepts all items and is
$1$-competitive.

\item $\B(\LG) > \R - 1/n$.
Let $N$ be the set of all non-large items accepted by $\ALG$. 
By \cref{lem:large-gain}, 
\begin{align*}
    \ALG(\sigma) 
    & = \gain(\LG) + \gain(N) 
    \gtrsim \fintegral(\B(\LG)) + \gain(\LG^{\geq x}) - x \cdot \B(\LG^{\geq x}) + \gain(N) \\
    & \eqsim R \cdot x + \gain(\LG^{\geq x}) - x \cdot \B(\LG^{\geq x}) + \gain(N).
\end{align*}  
where for the last relation we used $\fintegral(\B(\LG)) \eqsim \fintegral(\f^{-1}(x)) = \R \cdot x$ 
(by \cref{lem:F-properties}).

As \ALG takes all non-large items and all large items that are at least $x$, the
input sequence~$\sigma$ contains items taken by \ALG and possibly some large
items smaller than~$x$. Thus, the gain of \OPT on large items is maximized when
it takes $\LG^{\geq x}$ and fills the remaining $n - |L^{\geq x}|$ bins with
large items from $\sigma$ smaller than $x$. The total gain of \OPT
is thus at most 
\[
  \OPT(\sigma) \leq \gain(L^{\geq x}) + x \cdot (1-\B(\LG^{\geq x})) + \gain(N) 
  = x + \gain(L^{\geq x}) - x \cdot \B(\LG^{\geq x}) + \gain(N).
\]
Comparing the bounds on gains of \ALG and \OPT and observing that 
the term $\gain(L^{\geq x}) - x \cdot \B(\LG^{\geq x}) + \gain(N)$
is non-negative, yields $\ALG(\sigma) \geq \R \cdot \OPT(\sigma) - O(1/n)$. 
As $\OPT(\sigma) \geq \gain(\LG) = \Omega(1)$, we obtain
$\ALG(\sigma) \geq (\R - O(1/n)) \cdot \OPT(\sigma)$.
\qedhere
\end{itemize}
\end{proof}

As an immediate corollary, we observe that if $\sigma$ contains large
items only, then \ALG is $(R - O(1/n))$-competitive: If it terminates with empty
bins, then its competitive ratio follows by \cref{lem:empty}. Otherwise, it
terminates with $n$ large items, and hence, by \cref{lem:large-gain}, $\ALG(\sigma) \gtrsim
\fintegral(\B(\LG)) = \fintegral(1) = \R$. On the other hand, $\OPT(\sigma) \leq 1$, and
therefore the competitive ratio is at most $R - O(1/n)$ also in this case.

%%%%%%%%%%%%%%%%%%%%%%%%%%%%%%%%%%%%%%%%%%%%%%%%%%%%%%%%%%%%%%%%%%%%%%%%%%%%%%%%%%%%%%%%%%%%%%%%%
%%%%%%%%%%%%%%%%%%%%%%%%%%%%%%%%%%%%%%%%%%%%%%%%%%%%%%%%%%%%%%%%%%%%%%%%%%%%%%%%%%%%%%%%%%%%%%%%%

\section{Gain on medium items}
\label{sec:medium}

In the remaining part of the analysis, we make use of the following functions.
For any $c \in (1/2, 1]$, let $\cutintegral(c) = \int_0^1 \min \{ \f(y), c \}\,
\dd y$ and $\water(c) = \int_0^1 \max \{ c - \f(y), 0 \}\, \dd y$. Both
functions are increasing and depicted in \cref{fig:def-fpq} (right).
As we show below (cf.~the last property of \cref{lem:PQ-properties}),
$\cutintegral(c)$ lower-bounds the gain of \ALG in the case when
its load on non-\LT bins is at least $c$.

\begin{lemma}
\label{lem:PQ-properties}
Fix any $c \in (1/2, 1]$ and any $x \in [0,1]$. It holds that 
\begin{enumerate}
\item $\cutintegral(c) = c - \R \cdot c \cdot \ln (2 c)$,
\item $\water(c) = \R \cdot c \cdot \ln(2 c)$, 
\item $\cutintegral(c) + \water(c) = c$,
\item $\fintegral(x) + c \cdot (1 - x) \geq \cutintegral(c)$.
\end{enumerate}
\end{lemma}

\begin{proof}[Proof of \cref{lem:PQ-properties}]
  We fix any $c \in (1/2, 1]$ and any $x \in [0, 1]$.
  For the first property, observe that 
    \[ 
      \cutintegral(c) 
        = \int_0^{\f^{-1}(c)} \f(y) \, \dd y + \int_{\f^{-1}(c)}^1 c \, \dd y 
      = \fintegral(\f^{-1}(c)) + c \cdot (1-\f^{-1}(c)) 
      = c - \R \cdot c \cdot \ln (2 c),
  \]
  where for the last equality we used \cref{lem:F-properties}.
  Similarly, the second property follows as 
  \[
    \water(c) 
      = \int_0^{\f^{-1}(c)} c - \f(y) \, \dd y
      = c \cdot \f^{-1}(c) - \fintegral(\f^{-1}(c)) = \R \cdot c \cdot \ln(2 c) .
  \]
  The third relation, $\cutintegral(c) + \water(c) = c$, follows immediately 
  by the first two. Finally, for the last relation, we use
  \[
    \fintegral(x) + c \cdot (1 - x)
      = \int_0^{x} \f(y) \, \dd y + \int_{x}^1 c \, \dd y 
      \geq \int_0^{x} \min\{\f(y),c\} \, \dd x + \int_{x}^1 \min\{f(y), c\} \, \dd y 
      = \cutintegral(c).
  \]
  See also \cref{fig:F-vs-cutintegral} for a geometric argument.
  \end{proof}
    
\begin{figure}[t]
    \centering
    \includegraphics[width=0.7\textwidth]{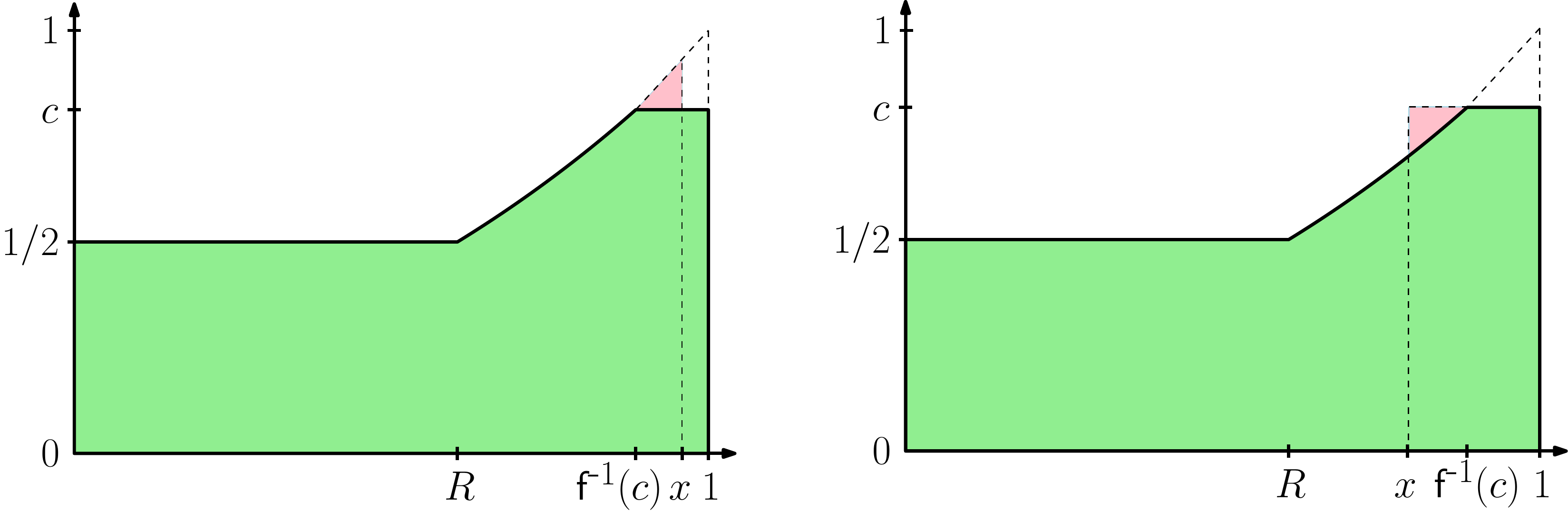}
    \caption{A geometric interpretation of the second property of \cref{lem:PQ-properties}:
    using $\cutintegral$ to lower-bound the sum of $\fintegral(x)$ and the rectangle $c \cdot (1-x)$
      for the case $\f^{-1}(c) \leq x$ (left) and $\f^{-1}(c) > x$ (right).}
    \label{fig:F-vs-cutintegral}
\end{figure}

%%%%%%%%%%%%%%%%%%%%%%%%%%%%%%%%%%%%%%%%%%%%%%%%%%%%%%%%%%%%%%%%%%%%%%%%%%%%%%

\subsection{Boundary conditions on function \texorpdfstring{\g}{xi}}
\label{sec:boundary}

We start with a shorthand notation. Let $\T(a,b)  =  (a+b-1/2) \cdot (\g(b) - \g(a))$,
where $a, b \in [\smallBoundary, 1/2]$ (so that the values of $\g(a)$ and
$\g(b)$ are well defined).

Our choice of function $\g$ satisfies the conditions below. In fact, for our
analysis to hold, function $\g$ could be replaced by any Lipschitz continuous
and non-increasing function mapping $[\smallBoundary, 1/2]$ to $[0,1]$
satisfying these properties.

\begin{lemma}
\label{lem:g-properties}
The following properties hold for function $\g$:
\begin{enumerate}
  \item $x \cdot \g(x)$ is a non-increasing function of $x \in [\smallBoundary, 1/2]$,
    \label{item:g-area}
  \item 
    $\cutintegral(\stack(x)) + x \cdot \g(x) \geq \R$ for $x \in [\smallBoundary,1/2]$,
    \label{item:g-prop-0}
  \item 
    $\cutintegral(1-\smallBoundary) + 2 \smallBoundary \cdot \g(2 \smallBoundary) \geq \R$,
    \label{item:g-prop-phi}
  \item $2/3 - (2/3 - \smallBoundary) \cdot \g(x) \geq \R$ for $x \in [\smallBoundary, 1/3]$,
    \label{item:g-upper-bound}
  \item
    $\cutintegral(1-\smallBoundary) + \water(1-x) + (x+\smallBoundary-1) \cdot \g(x) + \max\{T(x,y),0\}
    \geq \R$
      for $x \in [1/3, 1/2]$ and $y \in [\smallBoundary, 2 \smallBoundary]$,
    \label{item:g-2dim-1}
  \item 
    $\cutintegral(\stack(y)) + \water(1-x) + (x - \stack(y)) \cdot \g(x) + \T(x,y) \geq \R$
    for $x \in [1/3, 1/2]$ and $y \in [\smallBoundary, x]$,
    \label{item:g-2dim-2}
  \item 
    $\cutintegral(2x) + \water(1-x)  - x \cdot \g(x) \geq \R$ for $x \in [1/3,1/2]$.
    \label{item:g-prop-1}
\end{enumerate}
\end{lemma}

%%%%%%%%%%%%%%%%%%%%%%%%%%%%%%%%%%%%%%%%%%%%%%%%%%%%%%%%%%%%%%%%%%%%%%%%%%%%%%

\subsection{Marked and tight items}

We start with a simple bound on the gain of \ALG on \M-bins. Recall that these
bins store single marked items.

\begin{lemma}
\label{lem:M-gain}
If \ALG terminates with at least one \M-bin, then $\gain(\M) \geq \mn(\M) \cdot \B(\M)$, 
and $\B(\M) \leq \g(\mn(\M))$. 
\end{lemma}

\begin{proof}
The first condition follows trivially as each \M-bin contains a single medium
item of size at least $\mn(\M)$. For the second condition, note that $\M \subseteq \D$, 
and thus also $\M \subseteq \D^{\geq \min(\M)}$. As \D is
\g-dominated, $\B(\M) \leq \B(D^{\geq \min(\M)}) \leq \g(\min(\M))$.
\end{proof}

We now take a closer look at the marked items and their influence on the gain on
other sets of items. We say that a medium marked item $x \in \D$ is
\emph{tight} if it is on the verge of violating $\g$-domination invariant.

\begin{definition}
\label{def:tight}
An item $x \in \D$ is tight if $\B(\D^{\geq x}) > \g(x) - 1/n$.
\end{definition}

If an item $x \in \D$ is tight, then another item of size $x$ or greater cannot
be included in~$\D$ without violating $\g$-domination invariant.
\cref{fig:tightness} (left) illustrates this concept. As $\D$ can only grow,
once an item becomes tight, it remains tight till the end. We emphasize that
items smaller than $x$ are not relevant for determining whether $x$ is tight. If
$\D$ contains a~tight item, then $\mt$ denotes the size of the minimum tight
item in~$\D$. This important parameter influences the gain both on set $D$ and
also on stacking bins $\MT$ and $\ST$.

\begin{lemma}
\label{lem:D-gain}
If $\D$ contains a tight item, then $\gain(\D) \gtrsim \mt \cdot \g(\mt)$.
\end{lemma}

\begin{proof}
Fix a tight item $d \in \D$ of size $\mt$. By \cref{def:tight},
$\B(\D^{\geq d}) > \g(d) - 1/n$, and thus $\gain(\D) \geq \gain(\D^{\geq d}) \geq d
\cdot \B(\D^{\geq d}) \gtrsim d \cdot \g(d)$.
\end{proof}

%%%%%%%%%%%%%%%%%%%%%%%%%%%%%%%%%%%%%%%%%%%%%%%%%%%%%%%%%%%%%%%%%%%%%%%%%%%%%%%

\subsection{Impact of tight items on stacking bins}

By \cref{item:g-area} of \cref{lem:g-properties}, $x \cdot \g(x)$ is a non-increasing 
function of $x$. Therefore, the smaller $\mt$ is, the larger is the lower bound
on $\gain(\D)$ guaranteed by \cref{lem:D-gain}. Now we argue that the larger $\mt$ is,
the better is the gain on stacking bins $\MT$ and $\ST$.

\begin{lemma}
\label{lem:rej-min-relation}
Assume \ALG failed to mark a medium item $y$.
Then, a tight item exists and $\mt \lesssim y$.
\end{lemma}

\begin{proof}
Let $\Dext = \D \cup \{y\}$. By the lemma assumption, \Dext is not
$\g$-dominated, i.e., there exists an item $x \in \Dext$ such that
$\B(\Dext^{\geq x}) > \g(x)$. Note that $x \leq y$, as otherwise we would
have $\D^{\geq x} = \Dext^{\geq x}$, and thus $\B(\D^{\geq x}) >
\g(x)$, which would contradict $\g$-domination of $D$.

Let $d \geq x$ be the minimum size of an item from $\D^{\geq x}$. 
Then, $\D^{\geq d} = \D^{\geq x}$, and thus 
\begin{equation}
\label{eq:d-is-tight} 
  \B(\D^{\geq d}) = \B(\D^{\geq x})
  = \B(\Dext^{\geq x}) - 1/n > \g(x) - 1/n \geq \g(d) - 1/n, 
\end{equation}
where in the last inequality we used monotonicity of $\g$. By
\eqref{eq:d-is-tight}, $d$ is tight. On the other hand, $\g$-domination of $\D$
implies that $\B(\D^{\geq d}) \leq \g(d)$. This, combined with
\eqref{eq:d-is-tight}, yields $\g(d) \eqsim \g(x)$, and thus $d \eqsim x \leq
y$. Note that $d$ remains tight till the end of the execution. This concludes
the lemma, as the minimum tight item, $\mt$, can be only smaller than $d$. 
\end{proof}

\begin{lemma}
\label{lem:mtD-min-stacked}
If \ALG finishes 
\begin{itemize}
\item with at least one \MT-bin, then $\mt$ is defined and
$\mt \lesssim \mn(\MT)$;
\item with at least one
\ST-bin, then $\mt$ is defined and $\mt \lesssim 2 \smallBoundary$.
\end{itemize}
\end{lemma}

\begin{proof}
For the first part of the lemma, fix a medium item from \MT of size
$\mn(\MT)$. By the definition of \ALG, it failed to mark this item.
Hence, by \cref{lem:rej-min-relation}, $\mt$ is defined and $\mt \lesssim
\mn(\MT)$.

Assume now that \ALG finishes with at least one \ST-bin. When the first such
\ST-bin was created, \ALG placed a small item $s < \smallBoundary$ in the
already existing \A-bin of load $r < \smallBoundary$, and the merge action
failed, because \ALG failed to mark the resulting item of size $s+r$. Thus,
again by \cref{lem:rej-min-relation}, $\mt$ is defined and $\mt \lesssim s+r
< 2\smallBoundary$.
\end{proof}

%%%%%%%%%%%%%%%%%%%%%%%%%%%%%%%%%%%%%%%%%%%%%%%%%%%%%%%%%%%%%%%%%%%%%%%%%%%%%%%

\begin{figure}[t]
\centering
\includegraphics[width=0.9\textwidth]{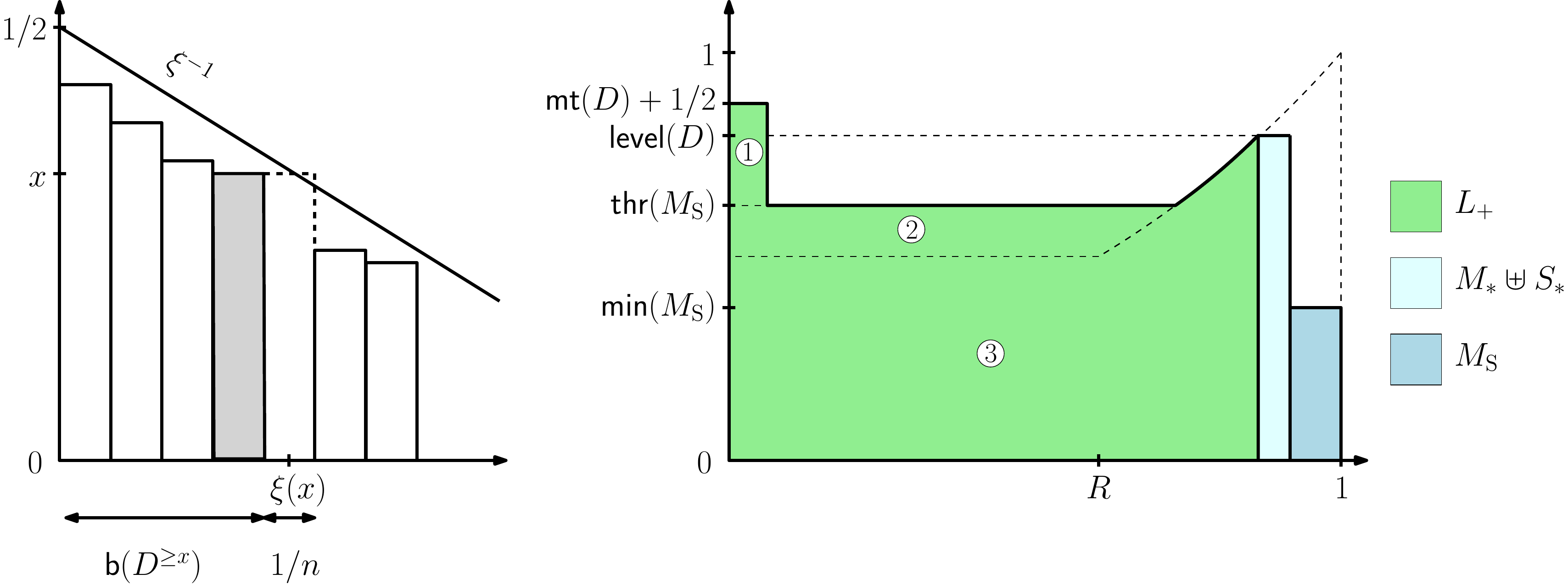}
\caption{Left: set $\D$ of marked items with a tight (gray) item of size $x$. As
$x$ is tight, insertion of another item of size~$x$ (with a dashed border) would
violate $\g$-domination of $\D$. Right: items collected by \ALG, when it
terminates without empty bins and with $\M$-bins. Gain on $\LT$-sets is split
into three parts, where the first part corresponds to the gain of $\Tmax$
(cf.~\cref{lem:gain-L-minus-T}). The minimum guaranteed load 
in $\M$-bins is given by \cref{lem:M-gain} and in the bins of $\MT
\uplus \ST$ by \cref{lem:stacking-gain}.}
\label{fig:tightness}
\end{figure}

To estimate the gain on \MT-bins and \ST-bins, we define
\begin{equation}
\label{eq:level-def}
  \level = \begin{cases}
  \min\{ \stack(\mt), 1-\smallBoundary \} & \text{if \D contains a tight item and $\ST \neq \emptyset$,} \\
  \stack(\mt) & \text{if \D contains a tight item and $\ST = \emptyset$,} \\
  1 & \text{if \D does not contain any tight item.}
\end{cases}
\end{equation}

\begin{lemma}
\label{lem:stacking-gain}
It holds that
$\gain(\MT \uplus \ST) \gtrsim \level \cdot \B(\MT \uplus \ST)$.
\end{lemma}

\begin{proof}
If \D does not contain a tight item, then, by \cref{lem:mtD-min-stacked}, both
\MT and \ST are empty, and the lemma follows trivially. Thus, in the
following we assume that \D contains a~tight item and we take a closer look at
the contents of \ST-bins and \MT-bins.

Assume that \MT is non-empty. \ALG creates a new $\MTI{i}$-bin (for $i \in \{2,3,4\}$) 
if the incoming medium item of category $\MTI{i}$ (of size $(1/(i+1),1/i]$) 
does not fit in any of the existing $\MTI{i}$ bins. 
Hence, each $\MTI{i}$-bin (except at most one) has exactly $i$ items, and therefore its
load is greater than $i/(i+1) \geq 2/3$ and is also at least $i \cdot \min(\MTI{i}) 
\geq 2 \cdot \min(\MT)$. Thus, 
\begin{equation}
\label{eq:gain-m-star}
  \gain(\MT) 
  \gtrsim \max \{2/3, 2 \cdot \mn(\MT) \} \cdot \B(\MT) = 
  \stack(\mn(\MT)) \cdot \B(\MT) \gtrsim \stack(\mt) \cdot \B(\MT),
\end{equation}
where the second inequality follows by \cref{lem:mtD-min-stacked} and by
monotonicity of function \stack. Note that \eqref{eq:gain-m-star} holds
trivially also when there are no \MT-bins.

If there are no \ST-bins, then, $\gain(\MT \uplus \ST) = 
\gain(\MT) \gtrsim \stack(\mt) \cdot \B(\MT) = \level \cdot \B(\MT \uplus \ST)$,
and the lemma follows. 

If there are some \ST-bins, recall that \ALG creates a new \ST-bin only if
the considered small item does not fit in any existing \ST-bin. Thus, the load
of each \ST-bin (except at most one) is at least~$1-\smallBoundary$, and
therefore $\gain(\ST) \gtrsim (1-\smallBoundary)\cdot \B(\ST)$. Combining this
with \eqref{eq:gain-m-star} implies $\gain(\MT \uplus \ST) \gtrsim \stack(\mt)
\cdot \B(\MT) + (1-\smallBoundary)\cdot \B(\ST) \geq \min \{\stack(\mt),
1-\smallBoundary \} \cdot \B(\MT \uplus \ST)$.
\end{proof}

%%%%%%%%%%%%%%%%%%%%%%%%%%%%%%%%%%%%%%%%%%%%%%%%%%%%%%%%%%%%%%%%%%%%%%%%%%%%%%%

\subsection{When RTA terminates without empty bins and without \texorpdfstring{\M-bins}{MS-bins}}
\label{sec:no-m-bin}

Using tight items, we may analyze the case when \ALG terminates without empty
bins and without $\M$-bins, and show that in such case its gain is approximately
greater than $\R$. As the gain of \OPT is at most $1$, this yields the desired
competitive ratio. 

\begin{lemma}
\label{lem:bin-count}
If \ALG terminates without empty bins, then $\B(\LG) + \B(\M) + \B(\MT) + \B(\ST) \eqsim 1$. 
\end{lemma}

\begin{proof}
There is at most one \A-bin.
The remaining bins (at least $n-1$ many) are of classes \LT, \M, $\MT$ or \ST,
and thus $\B(\LG) + \B(\M) + \B(\MT) + \B(\ST) \eqsim 1$. 
\end{proof}

\begin{lemma}
\label{lem:no-m-bin}
If on input $\sigma$, \ALG terminates without empty bins and without \M-bins, then $\ALG(\sigma) \gtrsim \R$.
\end{lemma}

\begin{proof}
We analyze the gain of \ALG on three disjoint sets: $\LG$, $D$ and $\MT \uplus \ST$.
\begin{align*}
  \ALG(\sigma)
  & \geq \gain(\LG) + \gain(\MT \uplus \ST) + \gain(\D) \\
  & \gtrsim \fintegral(\B(\LG)) + \level \cdot (1-\B(\LG)) + \gain(\D)
    &&& \text{(by L.~\ref{lem:large-gain}, L.~\ref{lem:stacking-gain}
      and L.~\ref{lem:bin-count})} \\
  & \gtrsim \cutintegral(\level) + \gain(\D) 
  &&& \text{(by L.~\ref{lem:PQ-properties})}
  \end{align*}
If $\D$ does not contain a tight item, then $\level = 1$,
and thus $\ALG(\sigma) \gtrsim \cutintegral(1) = \R$.

If $\D$ contains a tight item, then 
by \cref{lem:D-gain}, $\gain(\D) \geq \mt \cdot \g(\mt)$, and therefore
$\ALG(\sigma) \gtrsim \cutintegral(\level) + \mt \cdot \g(\mt)$.
We consider two cases.
\begin{itemize}
\item If $\level \geq \stack(\mt)$, 
then $\ALG(\sigma) \gtrsim \cutintegral(\stack(\mt)) + \mt \cdot \g(\mt) \geq \R$,
where the last inequality follows by \cref{item:g-prop-0} of \cref{lem:g-properties}.

\item 
The opposite case, $\level < \stack(\mt)$, is possible only if $\ST$-bins exist
and $\level = 1-\smallBoundary$. By \cref{lem:mtD-min-stacked}, the existence of
$\ST$-bins implies $\mt \lesssim 2 \smallBoundary$. As the function $x \cdot
\g(x)$ is non-increasing (cf.~\cref{item:g-area} of \cref{lem:g-properties}),
$\ALG(\sigma) \gtrsim \cutintegral(1-\smallBoundary) + 2 \smallBoundary \cdot
\g(2 \smallBoundary) \geq R$. The last inequality follows by
\cref{item:g-prop-phi} of \cref{lem:g-properties}.
\qedhere
\end{itemize}
\end{proof}

%%%%%%%%%%%%%%%%%%%%%%%%%%%%%%%%%%%%%%%%%%%%%%%%%%%%%%%%%%%%%%%%%%%%%%%%%%%%%%%%%%%%%%%%%%%%%%%%%
%%%%%%%%%%%%%%%%%%%%%%%%%%%%%%%%%%%%%%%%%%%%%%%%%%%%%%%%%%%%%%%%%%%%%%%%%%%%%%%%%%%%%%%%%%%%%%%%%

\section{Gain on large items revisited}
\label{sec:large2}

In this section, we assume that $\ALG$ terminates without empty bins and with at least one $\M$-bin.
Recall that \cref{lem:large-gain} allows us to estimate $\gain(\LT)$ by calculating the 
gain on \emph{large} items alone. Now we show how to improve this bound by
taking into account non-large items in $\LT$. First, we leverage the fact that if a
(marked) medium item is in $\M$, then \ALG must have failed to combine it with a large
item, and we obtain a better lower bound on the size of each large item. 
Second, we show that in some cases marked medium items must be in $\LT$ which increases its load.
If an \M-bin exists, we define
\begin{align*}
  &\Tmax 
  = \begin{cases}
    \T(\mn(\M), \mt) & \text{if $\mt$ is defined and 
        $\mn(\M) > \max\{ \mt, 1/3 \}$} \\
    0 & \text{otherwise}, 
  \end{cases} \\
  & \thresh = \min \{ 1-\mn(\M), 1/2+\smallBoundary \}.
\end{align*}

Note that $\Tmax$ is always non-negative. In particular, 
$\T(\mn(\M), \mt) = (\mn(\M) + \mt - 1/2) \cdot (\g(\mt) -
\g(\mn(\M)) \geq 0$ because $\mn(\M) + \mt \geq 1/2$ for $\mn(\M) > 1/3$ and $\mt \geq \smallBoundary$.

\begin{lemma}
\label{lem:lbin-mwd-full}
Assume \ALG terminates with at least one \M-bin.
Then, the load of any \LT-bin is at least $\thresh$.
\end{lemma}

\begin{proof}
Consider any \LT-bin $b$ and let $y$ be the large item contained in this bin. 
If $b$ contains an additional medium item, then its load is 
greater than $1/2 + \smallBoundary$, and the lemma follows. Hence, in the following, we assume 
that $y$ was not combined with a~medium item in a single bin. 
As \ALG finishes with an \M-bin, we fix a medium item $x$ of size $\mn(\M)$.
We consider three cases.
\begin{itemize}
\item Item $x$ arrived (or was created by merging some small items) before the
arrival of $y$. \ALG did not place $y$ in the $\M$-bin containing $x$,
because $y+x > 1$. Thus, $\load(b) \geq y \geq 1 - x = 1-\min(\M)$.

\item Item $x$ arrived after the arrival of $y$.
(Some small items might be placed together with $y$ prior to the arrival of $x$.)
As \ALG placed $x$ in a separate bin, it did not fit in $b$, i.e., 
the load of $b$ at the time of the arrival of $x$ was greater than $1 - x = 1-\min(\M)$.

\item Item $x$ was created by merging small items after the arrival of $y$.
Let $s < x$ be the small item that caused the creation of $x$. \ALG placed $s$ in \A-bin, because 
$s$ did not fit in $b$, i.e., the load of $b$ at that time
was greater than $1 - s > 1 - x = 1-\min(\M)$.
\qedhere
\end{itemize}
\end{proof}

\begin{lemma}
\label{lem:gain-L-minus-T}
Assume that \ALG terminates with at least one $\M$-bin.
Then, $\gain(\LT) \gtrsim \Tmax + \int_0^{\B(\LG)} \max\{ \f(y), \thresh \} \, \dd y$.
\end{lemma}

\begin{proof}
We sort accepted large items by their arrival time and denote the bin containing
the $i$-th large item by $b_i$. The bin $b_i$ contains a large item of size at
least $\f(i/n)$ because of the acceptance threshold, and its load is at least
$\thresh$ by \cref{lem:lbin-mwd-full}, i.e., $\load(b_i) \geq \max \{ \f(i/n),
\thresh \}$.

We now show how to decrease the load in $\LT$-bins, so that the
remaining load in bin~$b_i$ remains at least $\max \{ \f(i/n), \thresh \}$ 
and the change in the total gain is approximately equal to $\Tmax$. 
This claim
is trivial for $\Tmax = 0$, so we assume $\Tmax > 0$. This
is possible only if a tight item exists, $\mn(\M) > \mt$ and $\mn(\M) > 1/3$. As
$\mn(\M) > \mt$, every marked medium item of size from the interval $[\mt, \mn(\M))$ is in (a
separate) \LT-bin; let $\tilde{L}$ be the set of these bins. As $\M \subseteq
\D^{\geq \mt}$, $n \cdot \B(\tilde{L}) = |\D^{\geq \mt} \setminus \M| =
|\D^{\geq \mt}| - |\M|$. Using the tightness of $\mt$ and \cref{lem:M-gain},
$\B(\tilde{L}) = \B(\D^{\geq \mt}) - \B(\M) \gtrsim \g(\mt) - \g(\mn(\M))$. From
each bin of $\tilde{L}$ we remove a~load of $\mt + \mn(\M) - 1/2$.
The induced change in the total gain is then approximately equal to $\B(\tilde{L}) \cdot (\mt + \mn(\M) - 1/2) =
\Tmax$.

We now analyze the load of bin $b_i$ after the removal. 
The original load of bin $b_i$ was at least $\f(i/n) + \mt$, and 
after removal it is at least 
$\f(i/n) + \mn(\M) - 1/2$. This amount is at least $\f(i/n)$ (as $\mn(\M) \leq 1/2$)
and at least $1-\mn(\M) \geq \thresh$ (as $\f(i/n) \geq 1/2$).
Hence, the remaining load of $b_i$ is at least $\max \{ \f(i/n), \thresh \}$. 

Thus,
$\gain(\LT) - \Tmax 
  \geq \frac{1}{n} \sum_{i=1}^{n \cdot \B(\LG)} \max\{ \f(i/n), \thresh \}
  \eqsim \int_0^{\B(\LG)} \max\{ \f(y), \thresh \} \, \dd y$. 
Particular subsets of $\LT$ are depicted in \cref{fig:tightness} (right);
the removed part of gain $\Tmax$ is depicted as~$(1)$.
\end{proof}

\begin{lemma}
\label{lem:gain-L-improved}
Assume that \ALG terminates with at least one $\M$-bin.
Then, $\gain(\LT) + \eta \cdot (1-\B(\LG)) \gtrsim \cutintegral(\eta) +
\water(\min\{\thresh, \eta\}) + \Tmax$ for any $\eta \in (1/2, 1]$.
\end{lemma}

\begin{proof}
We fix any $\eta \in (1/2,1]$ and define $h = \min\{ \thresh, \eta \}$. By
\cref{lem:gain-L-minus-T}, $\gain(\LT) - \Tmax + \eta \cdot (1-\B(\LG)) \gtrsim
\int_0^{\B(\LG)} \max\{ \f(y), h \} \, \dd y + \int_{\B(\LG)}^1 \eta \,\dd y$.
We denote this lower bound by $A(\B(\LG))$ and we analyze it as a function of
$\B(\LG)$. When $\B(\LG) < \f^{-1}(h)$, then using $h \leq \eta$ we obtain
$A(\B(\LG)) = \int_0^{\B(\LG)} h \, \dd y + \int_{\B(\LG)}^1 \eta \,\dd y \geq
\int_0^{\f^{-1}(h)} h \, \dd y + \int_{\f^{-1}(h)}^1 \eta \,\dd y =
A(\f^{-1}(h))$. Therefore, we need to lower-bound the value of $A(\B(\LG))$ only
for $\B(\LG) \geq \f^{-1}(h)$. In such case, 
\begin{align*}
  A(\B(\LG)) 
  & = \int_0^{\B(\LG)} \max\{ h, \f(y)\} \, \dd y + \int_{\B(\LG)}^1 \eta \,\dd y \\
  & = \int_0^{\B(\LG)} \f(y) \, \dd y + \int_{\B(\LG)}^1 \eta \,\dd y + \int_0^{\B(\LG)} \max\{ h - \f(y), 0 \} \, \dd y \\
  & \geq \int_0^1 \min \{ \f(y), \eta \} + \int_0^1 \max\{ h - \f(y), 0 \} \, \dd y 
  = \cutintegral(\eta) + \water(h). 
\qedhere
\end{align*}
\end{proof}

\subsection{When RTA terminates without empty bins and with some \texorpdfstring{\M-bins}{MS-bins}}

The following lemma combines our bounds on gains on $\LT$, $\MT$, $\ST$ and $\M$.

\begin{lemma}
\label{lem:helper-bound}
Assume that $\ALG$ run on input $\sigma$ terminates without empty bins and with at least one $\M$-bin.
Then, for any $\eta \in (1/2, \level]$,
\[
\ALG(\sigma) \gtrsim \cutintegral(\eta) + \water(\min\{\eta, \thresh\}) + \Tmax + 
(\mn(\M) - \eta) \cdot \g(\mn(\M)).
\]
\end{lemma}
  
\begin{proof}
By the lemma assumptions,
\begin{align*}
      \ALG(\sigma)
      & \geq \gain(\LT) + \gain(\MT \uplus \ST) + \gain(\M) \\
      & \gtrsim \gain(\LT) +  \eta \cdot \B(\MT \uplus \ST) + \mn(\M) \cdot \B(\M) 
        &&& \text{(by L.~\ref{lem:stacking-gain} and L.~\ref{lem:M-gain})} \\
      & \eqsim \gain(\LT) + \eta \cdot(1-\B(\LT))
        + (\mn(\M) - \eta) \cdot\B(\M). 
        &&& \text{(by L.~\ref{lem:bin-count})} 
\end{align*}
Applying the guarantee of \cref{lem:gain-L-improved} to $\gain(\LT) + \eta \cdot(1-\B(\LT))$
concludes the proof.
\end{proof}

\begin{lemma}
\label{lem:small-ms}
Assume that $\ALG$ run on input $\sigma$ terminates 
without empty bins and with at least one \M-bin.
If $\mn(\M) \leq 1/3$, then $\ALG(\sigma) \gtrsim \R$.
\end{lemma}

\begin{proof}
As $\level \geq 2/3$, we may apply \cref{lem:helper-bound} 
with $\eta = 2/3$. Note that $\thresh \geq 2/3$ for $\min(\M) \leq 1/3$. Then,
\begin{align*}
  \ALG(\sigma)
  & \gtrsim \cutintegral(2/3) + \water(2/3) + (\mn(\M) - 2/3) \cdot \g(\mn(\M)) \\
  & = 2/3 + (\smallBoundary - 2/3) \cdot \g(\mn(\M)) \geq \R.
      &&& \text{(by L.~\ref{lem:PQ-properties})} 
\end{align*}
The last inequality follows by \cref{item:g-upper-bound} of \cref{lem:g-properties}.
\end{proof}

\begin{lemma}
\label{lem:large-ms}
Assume that $\ALG$ run on input $\sigma$ terminates 
without empty bins and with at least one \M-bin.
If $\mn(\M) > 1/3$, then $\ALG(\sigma) \gtrsim \R$.
\end{lemma}

\begin{proof}
As $\mn(M) > 1/3$, $\thresh < 1-\mn(\M) \geq 2/3$. 
\cref{lem:helper-bound} applied with any $\eta \in [2/3, \level]$ yields
\begin{align}
\label{eq:large-ms}
  \ALG(\sigma)
  & \gtrsim \cutintegral(\eta) + \water(1-\mn(\M)) + \Tmax + (\mn(\M) - \eta) \cdot \g(\mn(\M)).
\end{align}
First, we assume that $\D$ has no tight items.
Then, $\level = 1$, and we may use \eqref{eq:large-ms} with $\eta = 2 \cdot \mn(\M)$ obtaining
$\ALG(\sigma) \geq \cutintegral(2\cdot\mn(\M)) + \water(1-\mn(\M)) - \mn(\M) \cdot \g(\mn(\M)) \geq \R$,
where the last inequality follows by \cref{item:g-prop-1} of \cref{lem:g-properties}.

Second, we assume that $\D$ contains a tight item and we consider three cases.
\begin{itemize}
\item $\level < \stack(\mt)$. 
This relation is possible only when $\level = 1-\smallBoundary$ and \ALG terminates with at
least one $\ST$-bin. 
In this case, \cref{lem:mtD-min-stacked} implies that $\mt \leq 2\smallBoundary$. 
We apply \eqref{eq:large-ms} with $\eta = 1-\smallBoundary$ obtaining
$\ALG(\sigma) \geq \cutintegral(1-\smallBoundary) + \water(1-\mn(\M)) + \Tmax + (\mn(\M) + \smallBoundary - 1) \cdot \g(\mn(\M))$, which is at least $\R$ by \cref{item:g-2dim-1} of \cref{lem:g-properties}. 

\item $\level \geq \stack(\mt)$ and $\mn(\M) \leq \mt$.
Using monotonicity of $\stack$, $\level \geq \stack(\mn(\M)) =
2\cdot \mn(\M)$. 
Applying \eqref{eq:large-ms} with $\eta = 2\cdot \mn(\M)$ yields
$\ALG(\sigma) \gtrsim \R$ by \cref{item:g-prop-1} of \cref{lem:g-properties}.

\item $\level \geq \stack(\mt)$ and $\mn(M) > \mt$.
In this case, $\Tmax = \T(\min(\M), \mt)$. 
Applying \eqref{eq:large-ms} with $\eta = \stack(\mt)$ yields
$\ALG(\sigma) \gtrsim \cutintegral(\stack(\mt)) + \water(1-\mn(\M)) + \T(\min(\M), \mt) + (\mn(\M) -
\stack(\mt) \cdot \g(\mn(\M))$, which is at least $\R$
by \cref{item:g-2dim-2} of \cref{lem:g-properties}.
\qedhere
\end{itemize}
\end{proof}

%%%%%%%%%%%%%%%%%%%%%%%%%%%%%%%%%%%%%%%%%%%%%%%%%%%%%%%%%%%%%%%%%%%%%%%%%%%%%%%%%%%%%%%%%%%%%%%%%
%%%%%%%%%%%%%%%%%%%%%%%%%%%%%%%%%%%%%%%%%%%%%%%%%%%%%%%%%%%%%%%%%%%%%%%%%%%%%%%%%%%%%%%%%%%%%%%%%

\section{Competitive ratio of RTA}
\label{sec:final}

\begin{theorem}
The competitive ratio of \ALG for the multiple knapsack problem is at least $\R-O(1/n)$. 
\end{theorem}
  
\begin{proof}
Fix an input $\sigma$. If $\ALG(\sigma)$ terminates 
with some empty bins, then its competitive ratio
follows by \cref{lem:empty}.

Hence, below we assume that \ALG terminates without empty bins.
We presented three lemmas that cover all possible cases:
there are no \M-bins (\cref{lem:no-m-bin}),
there are \M-bins and $\mn(\M) \leq 1/3$ (\cref{lem:small-ms}),
and there are \M-bins and $\mn(\M) > 1/3$ (\cref{lem:large-ms}).
In all these cases, we proved $\ALG(\sigma) \gtrsim \R$.
As $\OPT(\sigma)  \leq 1$, the theorem follows.
\end{proof}

%%%%%%%%%%%%%%%%%%%%%%%%%%%%%%%%%%%%%%%%%%%%%%%%%%%%%%%%%%%%%%%%%%%%%%%%%%%%%%%%%%%%%%%%%%%%%%%%%
%%%%%%%%%%%%%%%%%%%%%%%%%%%%%%%%%%%%%%%%%%%%%%%%%%%%%%%%%%%%%%%%%%%%%%%%%%%%%%%%%%%%%%%%%%%%%%%%%

\section{Proof of Lemma 7}
\label{sec:proofs}

We start with technical helper claims.

\begin{fact}
  The functions below are the derivatives of functions $\cutintegral$, $\water$ and $\g$, respectively.
  \begin{align*}
    \cutintegral'(x) & =  -\R \cdot \ln x \\
    \water'(x) & =  1+\R\cdot \ln x \\
    \g'(x) & =
    \begin{cases}
      -\gconst / x^2 & \text{if $x < 1/3$} , \\
      -18 \gconst & \text{if $x > 1/3$}
    \end{cases}
  \end{align*}
  \end{fact}

  \begin{lemma}
  \label{lem:T-non-increasing}
  $\T(x,y)$ is a non-increasing function of $y$ in the interval $[\smallBoundary,1/2]$.
  \end{lemma}
  
  \begin{proof}
  Recall that $\T(x,y) = (x+y-1/2) \cdot (\g(y) - \g(x))$ is defined for $x,y \in [\smallBoundary,1/2]$.
  As the function $\T(x,y)$ is continuous and differentiable everywhere except $y = 1/3$,
  it suffices to show that its partial derivative $\partial_y \T(x,y)$ is non-positive (except $y = 1/3$).
  We have 
  \[
    \partial_y \T(x,y) = \g(y) - \g(x) + (x+y-1/2) \cdot \g'(y) \leq \g(y) + y \cdot \g'(y),
  \]
  where for the inequality we used $\g(x) \geq 0$ and $(x-1/2) \cdot \g'(y) \geq 0$.
  
  If $y \leq 1/3$, then $\partial_y \T(x,y)
  \leq \gconst/y + y \cdot (- \gconst/y^2) = 0$. If $y > 1/3$, then $\partial_y
  \T(x,y) \leq 9 \gconst \cdot (1 - 2y) - 18 \gconst \cdot y = 9 \gconst \cdot (1-4y) <
  0$, which concludes the proof.
  \end{proof}
  
  %%%%%%%%%%%%

\begin{lemma}
  \label{lem:P2x_approx}
  $P(2/3) = \R - \gconst$ and 
  for any $y \in [2/3, 1]$ it holds that $\cutintegral(y) \geq \R + 3 \gconst \cdot y - 3 \gconst$.
  \end{lemma}
  
\begin{proof}
  For the first part of the lemma, observe that 
  by the definition of $\gconst$ (see \eqref{eq:gconst_def}), 
  \begin{align*}
    \cutintegral(2/3) + \gconst 
      & = 2/3 - \R \cdot (2/3) \cdot \ln (4/3) + 
        (1 + (2/3) \cdot \ln(4/3)) \cdot \R - 2/3 = \R.
  \end{align*}
  To show the second relation, note that 
  $\cutintegral(2/3) = \R - \gconst$ and $\cutintegral(1) = \R$. 
  Let $h(y) = R + 3 \gconst \cdot y - 3 \gconst$ 
  be the linear function that coincides with $\cutintegral(y)$ for $y = 2/3$ and $y = 1$.
  As the function $\cutintegral$ is concave on its whole domain (its second derivative 
  $\cutintegral''(y) = -R/y$ is negative), we have $\cutintegral(y) \geq h(y)$ for any $y \in [2/3,1]$.
\end{proof}

\begin{proof}[Proof of \cref{lem:g-properties}]

Note that $\g(1/3) = 3 \cdot \gconst$ and $\g(1/2) = 0$, 
For each property, we define an~appropriate function $G_i$ that we analyze;
for all properties except the first one, we show that the function 
value is at least $\R$ for an appropriate range of arguments. 

\begin{description}

\item[\cref{item:g-area}.]
Let $G_1(x) = x \cdot \g(x)$. 
Then $G_1(x) = \gconst$ for $x \in [\smallBoundary, 1/3]$, and 
for $x \in [1/3,1/2]$ it holds that $G_1(x) = 9 \gconst \cdot x \cdot (1-2x)$, i.e.,
the function $G_1(x)$ is decreasing. Hence, $G_1(x)$ is non-increasing in the whole domain 
$[\smallBoundary, 1/2]$.

\item[\cref{item:g-prop-0}.]
Let $G_2(x) = \cutintegral(\stack(x)) + x \cdot \g(x)$; we want to show that
$G_2(x) \geq \R$ for any $x \in [\smallBoundary,1/2]$. For $x \in
[\smallBoundary, 1/3]$, it holds that $G_2(x) = \cutintegral(2/3) + \gconst = R$
(by \cref{lem:P2x_approx}). For $x \in [1/3,1/2]$, $G_2(x) = \cutintegral(2x) +
x \cdot \g(x)$. Using \cref{lem:P2x_approx}, we obtain $G_2(x) \geq  \R + 6
\gconst \cdot x - 3 \gconst + 9 \gconst \cdot (x - 2x^2) = \R - 18 \gconst \cdot
(x - 1/2) \cdot (x - 1/3) \geq \R$. The last inequality follows as for any $x
\in [1/3,1/2]$ the term $(x - 1/2) \cdot (x - 1/3)$ is non-positive.

\item[\cref{item:g-prop-phi}.]
Let $G_3 = \cutintegral(1-\smallBoundary) + 2 \smallBoundary \cdot \g(2
\smallBoundary)$. It can be verified numerically that \mbox{$G_3 > 0.593 > \R$}.

\item[\cref{item:g-upper-bound}.] 
Let $G_4(x) = 2/3 - (2/3 - \smallBoundary) \cdot \g(x)$. 
As the function $\g(x)$ is decreasing, for any $x \in [\smallBoundary, 1/3]$ it holds that 
$G_4(x) \geq G_4(\smallBoundary) 
= 2/3 - (2/3) \cdot \gconst / \smallBoundary + \gconst$. Substituting the definition of $\smallBoundary$ 
(see \eqref{eq:small_boundary}), we obtain 
$G_4(x) \geq G_4(\smallBoundary) \geq 2/3 - 2/3 + \R - \gconst + \gconst = \R$. 

\item[\cref{item:g-2dim-1}.]
Let $G_5(x,y) = \tilde{G}_5(x) + \max \{\T(x,y), 0 \}$,
where $\tilde{G}_5(x) = \cutintegral(1-\smallBoundary) + \water(1-x) + (x + \smallBoundary - 1) \cdot \g(x)$.
We want to show that $G_5(x,y) \geq \R$ for any $x \in [1/3, 1/2]$
and $y \in [\smallBoundary, 2\smallBoundary]$. 
\begin{itemize}
\item If $x \in [1/3,2 \smallBoundary]$, then already $\tilde{G}_5(x) \geq \R$. To show this relation,
we estimate its derivative in the interval $[1/3, 2\smallBoundary]$:
\begin{align*}
  \tilde{G}_5'(x)
  & = -1 - \R\cdot \ln (1-x) + (x + \smallBoundary - 1) \cdot \g'(x) + \g(x) \\
  & = -1 - \R\cdot \ln (1-x) - 36 \gconst \cdot x - 18 \gconst \cdot \smallBoundary + 27 \gconst \\
  & \leq -1 - \R\cdot \ln (1-2 \smallBoundary) - 36 \gconst \cdot (1/3) - 18 \gconst \cdot \smallBoundary +
  27 \gconst < -0.2479 < 0.
\end{align*}
Hence, $\tilde{G}_5$ is decreasing in the interval $[1/3, 2 \smallBoundary]$, and thus for any 
$x \in [1/3, 2\smallBoundary]$ it holds that 
$G_5(x,y) \geq \tilde{G}_5(x) \geq \tilde{G}_5(2 \smallBoundary) > 0.5997 > \R$.

\item If $x \in (2 \smallBoundary, 1/2]$, then $\T(x,y) \geq 0$. 
As $\T(x,y)$ is a~non-increasing function of $y$ in the interval $[\smallBoundary, 2\smallBoundary]$ 
(by \cref{lem:T-non-increasing}), it holds that $\T(x,y) \geq \T(x,2\smallBoundary)$. Therefore,
\[
  G_5(x,y) \geq G_5(x, 2\smallBoundary)
    = \cutintegral(1-\smallBoundary) + \water(1-x) - (\smallBoundary + 1/2) \cdot \g(x) 
    + (x+2\smallBoundary-1/2) \cdot \g(2\smallBoundary).
\]
We now estimate its partial derivative for $x \in [2 \smallBoundary, 1/2]$:
\begin{align*}
  \partial_x G_5(x, 2\smallBoundary)
  & = -1 - \R \cdot \ln (1-x) - (\smallBoundary + 1/2) \cdot \g'(x) + \g(2\smallBoundary) \\ 
  & \leq -1 - \R \cdot \ln(1/2) + 18 \gconst \cdot (\smallBoundary + 1/2) + 9 \gconst \cdot (1 - 4 \smallBoundary) 
   < -0.0673 < 0.
\end{align*}
Therefore $G_5(x,2\smallBoundary)$ is decreasing as a function of $x$ in the
interval $[2\smallBoundary, 1/2]$. Thus, for the considered range of arguments,
$G_5(x,y) \geq G_5(x,2\smallBoundary) \geq G_5(1/2, 2\smallBoundary) > 0.5934 >
\R$. 
\end{itemize}

\item[\cref{item:g-2dim-2}.]
Let 
\[ 
  G_6(x,y) = \cutintegral(\stack(y)) +  \water(1-x) + (x - \stack(y)) \cdot \g(x) + \T(x,y)
\]
Fix any pair $(x,y) \in [1/3, 1/2] \times [\smallBoundary, 1/2]$, such that 
$y \leq x$. We prove that \mbox{$G_6(x,y) \geq \R$}, by showing each of the following
inequalities
\begin{equation}
\label{eq:g6_seq}
  G_6(x,y) \geq G_6(x,y') \geq \tilde{G}_6(x,y') \geq \tilde{G}_6(1/2,y'') \geq \R,
\end{equation}
where $y' = \max\{y, 1/3\}$ and $y'' = y'+1/2-x$. 
The function $\tilde{G}_6$ is a lower bound on function~$G_6$ created
by using \cref{lem:P2x_approx} (and defined formally later).

The first inequality of \eqref{eq:g6_seq} is trivial for $y \geq 1/3$, hence 
we assume that $y < 1/3$. In such case, using the definition of function $\stack$,
we get $G_6(x,y) = \cutintegral(2/3) + \water(1-x) + (x - 2/3) \cdot \g(x) + \T(x,y)$.
As $\T(x,y)$ is a decreasing function of $y$, we obtain that $\T(x,y) \geq \T(x,y')$ 
and thus also $G_6(x,y) \geq G_6(x,y')$.

To show the second inequality  of \eqref{eq:g6_seq}, we first simplify $G_6(x,y)$ 
using that $y \geq 1/3$ and the definition of $\stack$:
\begin{align*}
  G_6(x,y) 
  & = \cutintegral(2 y) + \water(1-x) + (x - 2y) \cdot \g(x) + (x+y-1/2) \cdot (\g(y) - \g(x)) 
\end{align*}
Now, using \cref{lem:P2x_approx}, we have
\begin{align*}
  G_6(x,y) 
  & \geq \R + \water(1-x) + 6 \gconst \cdot y - 3 \gconst + (1/2 - 3y) \cdot \g(x) 
    + (x+y-1/2) \cdot \g(y).
\end{align*}
We denote the estimate above by $\tilde{G}_6(x,y)$ and we inspect its 
two partial derivatives.
\begin{align*}
  \partial_x \tilde{G}_6(x,y) 
  & =  -1 - \R \cdot \ln(1-x) + (1/2 - 3y) \cdot \g'(x) + \g(y) \\
  \partial_y \tilde{G}_6(x,y) 
  & = 6 \gconst - 3 \cdot \g(x) + \g(y) + (x+y-1/2) \cdot \g'(y) 
\end{align*}
The directional derivative along the vector $\vecv = (1,1)$ is then equal to 
\begin{align*}
  \partial_\vecv \tilde{G}_6(x,y) 
  & = \partial_x \tilde{G}_6(x,y) + \partial_y \tilde{G}_6(x,y) \\
  & =  -1 - \R \cdot \ln(1-x) + 6 \gconst - 3 \cdot \g(x) + 2 \cdot \g(y) 
   -18 \gconst \cdot (x - 2y) \\
  & =  -1 - \R \cdot \ln(1-x) + 36 \gconst \cdot x - 3 \gconst \\
  & \leq  -1 - \R \cdot \ln(1/2) + 36 \gconst \cdot (1/2) - 3 \gconst 
    < -0.0322 < 0.
\end{align*}
This means that if we take any point $(x,y') \in [1/3, 1/2] \times [1/3,1/2]$,
where $y' \leq x$, and move along vector $\vecv$, to point $(1/2,y'') = (1/2, y + 1/2 - x)$,
the value of the function $\tilde{G}_6$ can only decrease. This concludes the 
proof of the third inequality of \eqref{eq:g6_seq}. 

To show the final inequality of $\eqref{eq:g6_seq}$, 
we fix any $y'' \in [1/3,1/2]$. Then 
$\tilde{G}_6(1/2,y'') 
  = \R + 6 \gconst \cdot y - 3 \gconst + y'' \cdot \g(y'') 
  = \R - 18 \gconst \cdot (y'' - 1/2) \cdot (y''-1/3) 
  \geq \R$,
where the last inequality holds as $(y'' - 1/2) \cdot (y''-1/3)$ is non-positive.

\item[\cref{item:g-prop-1}.]
Let $G_7(x) = \cutintegral(2 x) + \water(1-x) - x \cdot \g(x)$. 
For $x \in [1/3,1/2]$, it holds that $G_7(x) = G_6(x,x)$.
Hence $G_7(x) \geq \R$ for $x \in [1/3,1/2]$ follows by \cref{item:g-2dim-2}.
\qedhere
\end{description}
\end{proof}

%%%%%%%%%%%%%%%%%%%%%%%%%%%%%%%%%%%%%%%%%%%%%%%%%%%%%%%%%%%%%%%%%%%%%%%%%%%%%%%%%%%%%%%%%%%%%%%%%
%%%%%%%%%%%%%%%%%%%%%%%%%%%%%%%%%%%%%%%%%%%%%%%%%%%%%%%%%%%%%%%%%%%%%%%%%%%%%%%%%%%%%%%%%%%%%%%%%

\bibliographystyle{plainurl}
\bibliography{ref}

\begin{thebibliography}{10}

\bibitem{AlKhLa19}
Susanne Albers, Arindam Khan, and Leon Ladewig.
\newblock Improved online algorithms for knapsack and {GAP} in the random order
  model.
\newblock In {\em Proc. 22nd Approximation, Randomization, and Combinatorial
  Optimization. Algorithms and Techniques (APPROX/RANDOM)}, pages 22:1--22:23,
  2019.
\newblock \href {http://dx.doi.org/10.4230/LIPIcs.APPROX-RANDOM.2019.22}
  {\path{doi:10.4230/LIPIcs.APPROX-RANDOM.2019.22}}.

\bibitem{ABFLNE02}
Yossi Azar, Joan Boyar, Lene~M. Favrholdt, Kim~S. Larsen, Morten~N. Nielsen,
  and Leah Epstein.
\newblock Fair versus unrestricted bin packing.
\newblock {\em Algorithmica}, 34(2):181--196, 2002.
\newblock \href {http://dx.doi.org/10.1007/s00453-002-0965-6}
  {\path{doi:10.1007/s00453-002-0965-6}}.

\bibitem{BaImKK07}
Moshe Babaioff, Nicole Immorlica, David Kempe, and Robert Kleinberg.
\newblock A knapsack secretary problem with applications.
\newblock In {\em Proc. 11th Approximation, Randomization, and Combinatorial
  Optimization. Algorithms and Techniques (APPROX/RANDOM)}, pages 16--28, 2007.
\newblock \href {http://dx.doi.org/10.1007/978-3-540-74208-1_2}
  {\path{doi:10.1007/978-3-540-74208-1_2}}.

\bibitem{BKoKR14}
Hans{-}Joachim B{\"{o}}ckenhauer, Dennis Komm, Richard Kr{\'{a}}lovic, and
  Peter Rossmanith.
\newblock The online knapsack problem: Advice and randomization.
\newblock {\em Theoretical Computer Science}, 527:61--72, 2014.
\newblock \href {http://dx.doi.org/10.1016/j.tcs.2014.01.027}
  {\path{doi:10.1016/j.tcs.2014.01.027}}.

\bibitem{BoFaLN01}
Joan Boyar, Lene~M. Favrholdt, Kim~S. Larsen, and Morten~N. Nielsen.
\newblock The competitive ratio for on-line dual bin packing with restricted
  input sequences.
\newblock {\em Nordic Journal of Computing}, 8(4):463--472, 2001.

\bibitem{CheKha05}
Chandra Chekuri and Sanjeev Khanna.
\newblock A polynomial time approximation scheme for the multiple knapsack
  problem.
\newblock {\em SIAM Journal on Computing}, 35(3):713--728, 2005.
\newblock \href {http://dx.doi.org/10.1137/S0097539700382820}
  {\path{doi:10.1137/S0097539700382820}}.

\bibitem{CyJeSg16}
Marek Cygan, {\L}ukasz Je{\.z}, and Jir{\'{\i}} Sgall.
\newblock Online knapsack revisited.
\newblock {\em Theory of Computing Systems}, 58(1):153--190, 2016.
\newblock \href {http://dx.doi.org/10.1007/s00224-014-9566-4}
  {\path{doi:10.1007/s00224-014-9566-4}}.

\bibitem{HaKaMa15}
Xin Han, Yasushi Kawase, and Kazuhisa Makino.
\newblock Randomized algorithms for online knapsack problems.
\newblock {\em Theoretical Computer Science}, 562:395--405, 2015.
\newblock \href {http://dx.doi.org/10.1016/j.tcs.2014.10.017}
  {\path{doi:10.1016/j.tcs.2014.10.017}}.

\bibitem{HaKaMY19}
Xin Han, Yasushi Kawase, Kazuhisa Makino, and Haruki Yokomaku.
\newblock Online knapsack problems with a resource buffer.
\newblock In {\em Proc. 30th Int. Symp. on Algorithms and Computation (ISAAC)},
  pages 28:1--28:14, 2019.
\newblock \href {http://dx.doi.org/10.4230/LIPIcs.ISAAC.2019.28}
  {\path{doi:10.4230/LIPIcs.ISAAC.2019.28}}.

\bibitem{IwaTak02}
Kazuo Iwama and Shiro Taketomi.
\newblock Removable online knapsack problems.
\newblock In {\em Proc. 29th Int. Colloq. on Automata, Languages and
  Programming (ICALP)}, pages 293--305, 2002.
\newblock \href {http://dx.doi.org/10.1007/3-540-45465-9_26}
  {\path{doi:10.1007/3-540-45465-9_26}}.

\bibitem{IwaZha10}
Kazuo Iwama and Guochuan Zhang.
\newblock Online knapsack with resource augmentation.
\newblock {\em Information Processing Letters}, 110(22):1016--1020, 2010.
\newblock \href {http://dx.doi.org/10.1016/j.ipl.2010.08.013}
  {\path{doi:10.1016/j.ipl.2010.08.013}}.

\bibitem{Keller99}
Hans Kellerer.
\newblock A polynomial time approximation scheme for the multiple knapsack
  problem.
\newblock In {\em Proc. 3rd Approximation, Randomization, and Combinatorial
  Optimization. Algorithms and Techniques (APPROX/RANDOM)}, pages 51--62, 1999.
\newblock \href {http://dx.doi.org/10.1007/978-3-540-48413-4_6}
  {\path{doi:10.1007/978-3-540-48413-4_6}}.

\bibitem{Keller16}
Hans Kellerer.
\newblock Knapsack.
\newblock In Ming-Yang Kao, editor, {\em Encyclopedia of Algorithms}, pages
  1048--1051. Springer, 2016.

\bibitem{KePfPi04}
Hans Kellerer, Ulrich Pferschy, and David Pisinger.
\newblock {\em Knapsack Problems}.
\newblock Springer, 2004.

\bibitem{KeRaTV18}
Thomas Kesselheim, Klaus Radke, Andreas T{\"{o}}nnis, and Berthold
  V{\"{o}}cking.
\newblock Primal beats dual on online packing {LPs} in the random-order model.
\newblock {\em SIAM Journal on Computing}, 47(5):1939--1964, 2018.
\newblock \href {http://dx.doi.org/10.1137/15M1033708}
  {\path{doi:10.1137/15M1033708}}.

\bibitem{MarVer95}
Alberto Marchetti{-}Spaccamela and Carlo Vercellis.
\newblock Stochastic on-line knapsack problems.
\newblock {\em Mathematical Programming}, 68:73--104, 1995.
\newblock \href {http://dx.doi.org/10.1007/BF01585758}
  {\path{doi:10.1007/BF01585758}}.

\bibitem{NogSar05}
John Noga and Veerawan Sarbua.
\newblock An online partially fractional knapsack problem.
\newblock In {\em 8th Int. Symp. on Parallel Architectures, Algorithms, and
  Networks (ISPAN)}, pages 108--112, 2005.
\newblock \href {http://dx.doi.org/10.1109/ISPAN.2005.19}
  {\path{doi:10.1109/ISPAN.2005.19}}.

\bibitem{Pising99}
David Pisinger.
\newblock An exact algorithm for large multiple knapsack problems.
\newblock {\em European Journal of Operational Research}, 114(3):528--541,
  1999.
\newblock \href {http://dx.doi.org/10.1016/S0377-2217(98)00120-9}
  {\path{doi:10.1016/S0377-2217(98)00120-9}}.

\bibitem{Vaze17}
Rahul Vaze.
\newblock Online knapsack problem and budgeted truthful bipartite matching.
\newblock In {\em Proc. 36th IEEE Int. Conf. on Computer Communications
  (INFOCOM)}, pages 1--9, 2017.
\newblock \href {http://dx.doi.org/10.1109/INFOCOM.2017.8057223}
  {\path{doi:10.1109/INFOCOM.2017.8057223}}.

\end{thebibliography}

\begin{appendix}

%%%%%%%%%%%%%%%%%%%%%%%%%%%%%%%%%%%%%%%%%%%%%%%%%%%%%%%%%%%%%%%%%%%%%%%%%%%%%%%%%%%%%%%%%%%%%%%%%

\section{Pseudocode of RTA}
\label{sec:pseudocode}

The pseudocode of our algorithm \ALG for handling a single incoming item is presented in \cref{alg:rta}.
The right margin shows label transitions of affected bins. 

\begin{algorithm}[!t]
  \caption{\ALG pseudocode for handling an item $x$}\label{alg:rta}
  \begin{algorithmic}[0]
  \If {$x$ is large}
     \If {$x \geq \f(\B(\LG) + 1/n)$}
        \If {\M-bin $b$ with $x$ space left exists}
           \State put $x$ in $b$ \Comment $\M \to \LT$
        \Else
           \State put $x$ in an empty bin \Comment $\emptyBucket \to \LT$
        \EndIf
     \Else
        \State reject $x$
     \EndIf
  \EndIf

  \Statex
  \If {$x$ is medium}
     \State Let $i \in \{2,3,4\}$ be such that $x \in [1/(i+1),1/i)$
     \If {\LT-bin $b$ with $x$ space left exists}
        \State put $x$ in $b$ \Comment $\LT \to \LT$
     \ElsIf {$D \cup \{x\}$ is \g-dominated}
        \State $D \gets D \cup \{x\}$ \Comment mark $x$
        \State put $x$ in an empty bin \Comment $\emptyBucket \to \M$
     \ElsIf {$\MTI{i}$-bin $b$ with $x$ space left exists}
        \State put $x$ in $b$ \Comment $\MTI{i} \to \MTI{i}$
     \Else 
        \State put $x$ in an empty bin \Comment $\emptyBucket \to \MTI{i}$
     \EndIf
  \EndIf
 
  \Statex
  \If {$x$ is small}
     \If {\LT-bin $b$ with $x$ space left exists}
        \State put $x$ in $b$ \Comment $\LT \to \LT$
     \ElsIf {an \ST-bin $b$ with $x$ space left exists}
        \State put $x$ in $b$ \Comment $\ST  \to \ST$
     \ElsIf {an \A-bin $b$ exists}
        \State put $x$ in $b$ \Comment $\A \to \A$
        \If {$\load(b) \geq \smallBoundary$}
           \If {$D \cup \{\load(b)\}$ is $\g$-dominated}
              \State merge all items from $b$ into medium item $y$ of size $\load(b)$ 
              \State $D \gets D \cup \{ y \}$ 
                 \Comment mark $y$, $\A \to \M$  
           \Else 
              \State relabel $b$ to \ST \Comment $\A \to \ST$
           \EndIf
        \EndIf
     \Else
        \State put $x$ in an empty bin \Comment $\emptyBucket \to \A$
     \EndIf
  \EndIf
  
  \Statex
  \If {there are no empty bins left}
  \State terminate
  \EndIf
\end{algorithmic}
\end{algorithm}

%%%%%%%%%%%%%%%%%%%%%%%%%%%%%%%%%%%%%%%%%%%%%%%%%%%%%%%%%%%%%%%%%%%%%%%%%%%%%%%%%%%%%%%%%%%%%%%%%
%%%%%%%%%%%%%%%%%%%%%%%%%%%%%%%%%%%%%%%%%%%%%%%%%%%%%%%%%%%%%%%%%%%%%%%%%%%%%%%%%%%%%%%%%%%%%%%%%

\section{Upper bound}
\label{sec:upper}

To show the optimality of \ALG, we prove that the competitive ratio of any
deterministic algorithm for the multiple knapsack problem is at most $\R -
O(1/n)$. This shows that the term $O(1/n)$ occurring in the competitive
ratio of our deterministic algorithm $\ALG$ is inevitable.

Our adversarial strategy is parameterized with a non-decreasing sequence of $n +
1$ numbers~$s(i)$, such that $1/2 < s(1) \leq s(2) \leq s(3) \leq \ldots \leq
s(n) \leq s(n+1) = 1$ and proceeds in at most $n+1$ phases numbered from $1$. In
phase $i$, the adversary issues $n$ items, all of size $s(i)$, and stops
immediately once \DET accepts one of these items (places it into a~bin). If all
$n$ items of a phase are rejected by \DET, the adversary finishes the input
sequence. 

Let $j$ be the number of phases where $\DET$ accepts an item. As the size of
each item from the input is strictly greater than $1/2$, at most one item fits
in a single bin, and thus $j \leq n$. Then, total load of \DET is equal to
$\DET(j) = \sum_{i = 1}^{j} s(i)$. Note that in the final phase $j+1$, $n$ items
of size $s(j+1)$ are presented (but rejected by \DET). \OPT may accept all of
them obtaining the total load $n \cdot s(j+1)$. By comparing these two gains and
as $\DET$ may choose the value of $j$, we get that the competitive ratio of
$\DET$ is at most
\begin{equation}
\label{eq:upper-bound-ratio}
U := \max_{0 \leq j \leq n} \frac{\DET(j) / n}{\OPT(j) / n} = 
  \max_{0 \leq j \leq n} \frac{ (1/n) \cdot \sum_{i = 1}^{j} s(i)}{s(j+1)}.
\end{equation}

It remains to choose an appropriate sequence $s$ and compute $U$ using \eqref{eq:upper-bound-ratio}. 
As a warm-up, we show a simple
construction leading to a weaker upper bound of $\R$, which we later refine to
obtain the bound of $\R - O(1/n)$. We note that the bound of $\R$ given by Cygan
et al.~\cite{CyJeSg16} uses concepts similar to ours, although in their
construction, the adversary tries to control the number of bins that store
particular amounts of items, while our approach is more fine-grained, and the
adversary controls the contents of particular bins. 

For our both bounds we use the following sequence $s$ with the function $\f$
defined in \eqref{eq:f-definition}, parameter $\alpha$ that we define later and
infinitesimally small but positive $\epsilon$. 
\begin{equation}
  \label{eq:ultimate_s}
    s(i) = \begin{cases}
      \f \left( \frac{i-1}{n} + \alpha \right) + \epsilon 
        & \text{if $i \in \{ 1, \ldots, n \}$}, \\
      1 & \text{if $i = n+1$}.
    \end{cases}
\end{equation}
The role of $\epsilon$ is to ensure that $s(i) > 1/2$ for all $i$. 
We will neglect the load caused by $\epsilon$ in our proofs below;
such approach could be formalized by taking $\epsilon = c/n^2$
with a constant $c$ tending to zero.

%%%%%%%%%%%%%%%%%%%%%%%%%%%%%%%%%%%%%%%%%%%%%%%%%%%%%%%%%%%%%%%%%%%%%%%%%%%%

\subsection{Upper bound of R}

\begin{lemma}
\label{lem:upper-bound-weak}
The competitive ratio of any deterministic algorithm \DET for the multiple
knapsack problem is at most $\R$.
\end{lemma}

\begin{proof}
The adversary uses the strategy described above
with sequence $s$ defined by \eqref{eq:ultimate_s} parameterized
with $\alpha = 0$.
We fix any $j \in \{0, \ldots, n\}$ and observe that $s(j+1) \geq \f(j/n)$ 
(for $j = n$ we use $\f(1) = 1$).
By \eqref{eq:upper-bound-ratio}, \eqref{eq:Ff} and \cref{lem:F-properties}, the competitive ratio of \DET is 
upper-bounded by 
\[
  U = \frac{ \sum_{i = 1}^{j} s(i)/n }{s(j+1)}
  = \frac{ \frac{1}{n} \cdot \sum_{i = 1}^{j} \f \left( \frac{i-1}{n} \right)}
    {s(j+1)}
  \leq \frac{ \int_0^{j/n} \f (x) \, \dd x }{s(j+1)}
  \leq \frac{\fintegral(j/n)}{\f(j/n)}
  \leq \R.
\qedhere
\]
\end{proof}

%%%%%%%%%%%%%%%%%%%%%%%%%%%%%%%%%%%%%%%%%%%%%%%%%%%%%%%%%%%%%%%%%%%%%%%%%%%%

\subsection{Upper bound of R - O(1/n)}

By choosing parameter $\alpha = 1/(8n)$, 
we may improve the upper bound by a term of $\Theta(1/n)$.
As in the proof of \cref{lem:upper-bound-weak}, 
we approximate the sum of $s(i)$'s by an appropriate integral, but 
this time we also bound the error due to such approximation.

\begin{lemma}
\label{lem:upper-bound-technical}
For the sequence $s$ defined in \eqref{eq:ultimate_s} parameterized
with $\alpha = 1/(8n)$, and for any $j \in \{0,\ldots,n\}$,
it holds that 
$\sum_{i=1}^j s(i)/n \leq \fintegral(j/n)$.
Moreover, for $n \geq 3$, it holds that 
$\sum_{i=1}^n s(i)/n \leq \fintegral(1) - 1/(52 n)$.
\end{lemma}

\begin{proof}
Note that $\sum_{i = 1}^j s(i)/n
= \sum_{i = 0}^{j-1} s(i+1) / n
 = \sum_{i = 0}^{j-1} \frac{1}{n} \cdot f(i/n+\alpha )$.
On the other hand, 
$F(j/n) = \int_0^{j/n} f(x) \,\dd x = \sum_{i = 0}^{j-1} \int_{0}^{1/n} \f(i/n + x) \,\dd x$.
Therefore, $F(j/n) - \sum_{i = 1}^j s(i)/n = \sum_{i=0}^{j-1} D_i$, where 
\begin{equation}
\label{eq:def_D}
  D_i = \int_{0}^{1/n} \f\left( i/n + x \right) 
      - \f\left( i/n + \alpha \right) \,\dd x .
\end{equation}
Now we estimate $D_i$. We consider two cases.
\begin{enumerate}
\item  $i/n + \alpha < \R$. Then, $f(i/n+\alpha) = 1/2$ and 
$\f(i/n+x) \geq 1/2$ for any $x \in [0,1/n]$. Substituting these bounds 
to \eqref{eq:def_D} yields $D_i \geq 0$. 

\item $i/n + \alpha \geq R$.
In this case, $\f(i/n + \alpha) = (2 \e)^{i/n + \alpha - 1}$ and 
$\f(i/n + x) \geq (2 \e)^{i/n + x - 1}$ for any $x \in [0,1/n]$.
Thus, 
\begin{align*}
    D_i 
      & \geq \int_{0}^{1/n} (2 \e)^{i/n + x - 1} - (2 \e)^{i/n + \alpha - 1} \,\dd x 
       = (2 \e)^{i/n - 1} \cdot 
        \int_{0}^{1/n} (2 \e)^x - (2 \e)^\alpha \,\dd x .
\end{align*}
By the Taylor expansion of function $(2\e)^x$,
for any $\beta \in [0,\R]$ it holds that
$1 + \beta/R \leq (2 \e)^\beta \leq 1 + 2 \cdot \beta/\R$. Thus, 
\begin{align*}
    D_i & \geq (2 \e)^{i/n - 1} \cdot 
        \int_{0}^{1/n} \frac{x}{\R} - \frac{2 \alpha}{\R} \,\dd x 
      \geq (2 \e)^{i/n - 1} \cdot \left( \frac{1/n^2}{2 \R} - \frac{2 \alpha / n}{\R} \right) \\
      & \geq \frac{1}{2 \e} \cdot \frac{1}{4 \R \cdot n^2} 
      > \frac{1}{13 \cdot n^2}.
\end{align*}
As $D_i \geq 0$ in either case, we immediately get the first part
of the lemma. 
To show the second part, we set $j = n$ and 
observe that all integers $i \in \{\lceil n \cdot \R \rceil, \ldots, j-1\}$ 
are considered in the second case above.
Therefore $F(1) - \sum_{i=1}^n s(i)/n = \sum_{i=0}^{n-1} D_i \geq 
(n-\lceil n \cdot R \rceil)/ (13 \cdot n^2)$. 
For $n \geq 3$, it holds that $n-\lceil n \cdot R \rceil \geq n/4$, which
concludes the proof.
\qedhere
\end{enumerate}
\end{proof}

\begin{theorem}
\label{thm:upper-bound}
The competitive ratio of any deterministic algorithm \DET for the multiple knapsack problem 
is at most $\R - 1/(52 n)$.
\end{theorem}

\begin{proof}
The adversary uses the strategy described above
with sequence $s$ defined by \eqref{eq:ultimate_s} parameterized
with $\alpha = 1/(8n)$. We compute the upper bound $U$ on the competitive ratio
of \DET using \eqref{eq:upper-bound-ratio}.

We separately consider the case of small value of $n$.
If $n = 1$, then $s(1) = 1/2$ and $s(2) = 1$, and therefore 
$U = 1/2$. If $n = 2$, then $s(1) = s(2) = 1/2$ and $s(3) = 1$,
and thus again $U = 1/2$. In these cases, $U = 1/2 < R - 1/(52n)$, and 
the theorem follows. 

Below, we assume that $n \geq 3$.
We fix any $j \in \{0, \ldots, n\}$ and let $y = j/n$.
We consider two cases. 

\begin{enumerate}
\item $j < n$. 
Then, by \cref{lem:upper-bound-technical}, $\sum_{i=1}^j s(i)/n \leq \fintegral(y)$. 
Additionally, $s(j+1) = \f(y+\alpha) + \epsilon > \f(y + \alpha)$. 
Using \cref{lem:F-properties}, we obtain
\begin{equation}
\label{eq:ub-1}
U = \frac{ (1/n) \cdot \sum_{i = 1}^{j} s(i)}{s(j+1)}
  < \frac{\fintegral(y)}{\f(y + \alpha)} 
  = \frac{\fintegral(y)}{\f(y)} \cdot \frac{\f(y)}{\f(y + \alpha)} 
  = \min\{y, \R\} \cdot \frac{\f(y)}{\f(y + \alpha)}.
\end{equation}
It suffices to show that $U \leq R - \alpha/2 = R - 1/(16n)$. 
\begin{enumerate}
\item If $y < \R - \alpha/2$, then $U < y \cdot 1 < \R - \alpha/2$.
\item If $y \geq \R - \alpha/2$, then we upper-bound the fraction
$\f(y) \,/\, \f(y + \alpha)$. As $y + \alpha/2 \geq \R$,
\begin{equation}
\label{eq:ub-2}
  \frac{\f(y)}{\f(y+\alpha)} 
  \leq \frac{\f(y+\alpha/2)}{\f(y+\alpha)} 
  = \frac{(2\e)^{y+\alpha/2-1}}{(2\e)^{y+\alpha-1}} 
  = (2\e)^{-\alpha/2}
  \leq 1-\frac{\alpha}{2\R}.
\end{equation}
In the last inequality, we used the Taylor expansion of the function $(2\e)^x$.
Combining \eqref{eq:ub-1} with \eqref{eq:ub-2} yields 
$U \leq R - \alpha/2$ as desired.
\end{enumerate}

\item $j = n$.
By \cref{lem:upper-bound-technical}, 
$\sum_{i=1}^n s(i)/n \leq R - 1/(52n)$ and $s(n+1) = 1$.
Therefore, in this case, $U \leq R - 1/(52n)$.
\qedhere
\end{enumerate}
\end{proof}

\end{appendix}

\end{document}